\documentclass[11pt,a4paper]{article}
\usepackage[english]{babel}
\usepackage{amsmath,amsthm,amssymb,epsfig,latexsym}
\usepackage[colorlinks=true,linkcolor=blue,citecolor=magenta]{hyperref}
\usepackage{color}
\usepackage{ulem}
\usepackage{bm}
\usepackage[numbers,square,sort&compress]{natbib}

\newcommand{\so}{{\scriptscriptstyle \rm I}}
\newcommand{\st}{{\scriptscriptstyle \rm I\hspace{-1pt}I}}

\newcommand{\bu}{\bar u}
\newcommand{\bv}{\bar v}

\newcommand{\bt}{\bar t}
\newcommand{\bw}{\bar w}
\newcommand{\bz}{\bar z}

\newcommand{\be}[1]{\begin{equation}\label{#1}}
\newcommand{\ba}[1]{\begin{multline}\label{#1}}
\newcommand{\ee}{\end{equation}}
\newcommand{\ea}{\end{multline}}

\newtheorem{thm}{Theorem}[section]
\newtheorem{prop}{Proposition}[section]
\newtheorem{lemma}{Lemma}[section]
\newtheorem{rem}{Remark}[section]
\newtheorem{cor}{Corollary}[section]
\newtheorem{Def}{Definition}[section]

\def\qed{\hfill\nobreak\hbox{$\square$}\par\medbreak}
 \makeatletter
 \@addtoreset{equation}{section}
 \makeatother
 
\newcommand{\bea}{\begin{eqnarray}}
\newcommand{\eea}{\end{eqnarray}}

\def\F{{\mathcal{F}}}

\def\BB{{\mathbb{B}}}

\def\CC{{\mathbb{C}}}
\newcommand{\ZZ}{{\mathbb Z}}
\def\TT{{\mathbb{T}}}
\def\FF{{\rm F}}

\def\rvac{|0\rangle}

\def\EE{{\rm E}}
\def\FF{{\rm F}}

\def\TT{\mathsf{T}}

\def\tFF{\tilde{\rm F}}

\def\Ee{{\sf e}}

\def\r#1{\eqref{#1}}
\def\sk#1{\left(#1\right)}

\def\Pfp{{P}^+_f}
\def\Pfm{{P}^-_f}

\def\RDYAn{\mathcal{D}Y_{\rm R}(\mathfrak{gl}_{n})}

\def\DYSn{\mathcal{D}Y(\mathfrak{sl}_{n})}
\def\RDYSn{\mathcal{D}Y_{\rm R}(\mathfrak{sl}_{n})}

\def\DYg{\mathcal{D}Y(\mathfrak{g})}
\def\DXg{\mathcal{D}X(\mathfrak{g})}

\def\ggo{\mathfrak{g}}

\def\ot{\otimes}
\def\cF{{\sf F}}

\def\Xif#1#2#3#4#5#6#7#8
{\Xi^{0,1}_{#1,#2}
{\small\Big(\begin{array}{ccc}#3&#4&#5\\ #6&#7&#8 \end{array}\Big)}}

\def\Oml{\Omega^{L}}
\def\Omr{\Omega^{R}}

\def\vn{\varnothing}

\def\fb{\mathfrak{b}}
\def\hc{\mathsf{h}}

\def\Om{\Omega}

\def\ovs{\rho}
\def\PSR{\Gamma}
\def\rt{\mathsf{r}}
\def\frd{\xi}
\def\Ig{I_\ggo}

\def\cre{\eta}
\def\gln{\mathfrak{gl}_n}
\def\gga{\mathfrak{sl}_n}
\def\goN{\mathfrak{o}_N}
\def\gsp{\mathfrak{sp}_{2n}}
\def\ggc{\mathfrak{sp}_{2n}}
\def\ggb{\mathfrak{o}_{2n+1}}
\def\ggd{\mathfrak{o}_{2n}}
\def\hC{H}
\def\ago{\mathfrak{a}}
\def\nus{\nu}
\def\nusp{\varrho}
\def\sB{U_{\mathcal{B}}}
\def\DeltaD{\Delta^{{\rm (D)}}}
\def\coun{\varepsilon}
\def\counD{\varepsilon^{(D)}}
\def\bgam{\gamma}
\def\bvphi{\varphi}
\def\mus{\mu}
\def\Tzm{\mathfrak{T}}

\def\fg{\mathfrak{g}}

\def\fl{\mathfrak{l}}
\def\fs{\mathfrak{s}}
\def\fp{\mathfrak{p}}
\def\fo{\mathfrak{o}}

\def\cV{\mathcal{V}}

\def\zcent{\mathfrak{z}}

\def\epsi{\epsilon}

\def\BF{\mathbb{F}}

\def\Tb{\bm{T}}

\def\IE{\mathcal{E}}
\def\IF{\mathcal{F}}

\def\oUF{\overline{Y}_F}

\def\bbw{\bar{\bm{\omega}}}
\def\Bf{\mathbf{B}}

\def\sB{U}
\def\RDYg{\mathcal{D}Y_{\rm R}(\fg)}

\begin{document}

\vspace{12pt}

\begin{center}
\begin{LARGE}
{\bf Yangian Doubles and off-Shell Bethe Vectors}
\end{LARGE}

\vspace{10mm}

\begin{large}
A.~Liashyk${}^{a}$, S.~Pakuliak${}^{b}$ and E.~Ragoucy${}^{b}$
\end{large}

\vspace{10mm}

${}^a$ {\it Beijing Institute of Mathematical Sciences and Applications (BIMSA),\\
No. 544, Hefangkou Village Huaibei Town, Huairou District Beijing 101408, China}

\vspace{2mm}

${}^b$ {\it Laboratoire d'Annecy de Physique Théorique (LAPTh)\\ 
CNRS \& Université Savoie Mont Blanc\\
Chemin de Bellevue, BP 110, F-74941, Annecy-le-Vieux cedex, France}
\vspace{2mm}

E-mails: liashyk@bimsa.cn, pakuliak@lapth.cnrs.fr, ragoucy@lapth.cnrs.fr

\bigskip

\end{center}

\begin{abstract}
Off-shell Bethe vectors for a generic $\fg$ invariant integrable model
are constructed through the currents of the Yangian doubles of the 
classical series.  These off-shell Bethe vectors  
are shown to satisfy the defining properties which were used in \cite{LPR-RR} 
to prove the rectangular recurrence relations and verify the eigenvalue property
of the on-shell Bethe vectors in $\ggo$-invariant integrable models.   
\end{abstract}


\section{Introduction}

The algebraic Bethe ansatz for $\ggo$-invariant integrable models provides a
systematic framework for constructing common eigenvectors of commuting integrals
of motion. Originally developed for low-rank algebras, in particular for 
$\ggo=\mathfrak{gl}_2$ \cite{S22}, it allows not only the construction of on-shell
Bethe vectors but also the computation of scalar products, norms, and form factors.

Extending these results to higher-rank symmetries and to algebras different from
$\gln$ requires an explicit construction of off-shell Bethe vectors. Early
approaches based on the nested Bethe ansatz \cite{KR1,KR2,Res85,Res91} yield Bethe
equations but are not efficient in constructing Bethe vectors themselves. A
combinatorial construction using the trace formula was developed for $\gln$- and
$U_q(\gln)$-invariant models in \cite{TV, TV-Com}, and later extended to the $C$ and $D$
series in \cite{GR}. To our best knowledge
no trace formula exist for the $B$ series.

Since 2007, a new method based on Drinfeld’s 'new' realization of quantum affine
algebras \cite{Dnew} has been proposed \cite{EKhP}. It relies on projections onto
intersections of different types of Borel subalgebras and expresses off-shell Bethe vectors as
projections of ordered products of simple root currents. This approach was applied
to supersymmetric Yangian doubles in \cite{HLPRS17}, with computational techniques
developed in \cite{KhP}. Using the Ding--Frenkel isomorphism \cite{DF}, off-shell Bethe vectors
can be written as polynomials in the monodromy matrix entries $T_{i,j}(u)$, whose
recurrence relations coincide with those obtained from the trace formula
\cite{TV-Com}.

For the calculation of scalar products, explicit formulas are not as important as the recurrence
relations arising from the action of monodromy matrix elements on off-shell Bethe
vectors. These relations determine recurrence relations for the highest coefficients of scalar products \cite{HLPRS-SP-gl, LPR-SP}. 
It was shown in \cite{LPR-RR} that for generic $\ggo$-invariant
models of all classical series, these relations follow from a minimal set of
defining relations for off-shell Bethe vectors. The goal of this paper is to prove
that the Bethe vectors defined by currents and projections satisfy these relations, with
the normalization of \cite{LPR-RR} ensuring the cancelation of zeros caused by the
Serre relations \cite{E}. 

We employ two realizations of Yangian doubles of classical type: the $RTT$
realization \cite{Dnew, KhT-DY} based on $\ggo$-invariant
$R$-matrices \cite{Y,ZZ79}, and the current realization formulated in 
 \cite{D-FTINT,KhT-DY,JYL20}. 

The paper is organized as follows. Section~\ref{g-all} presents a uniform
description of Yangian doubles for all algebras of the classical series
in terms of the currents and in $RTT$ realization.  
The main theorem~\ref{main-thm} of this paper is formulated in the Section~\ref{BVsect}.
It states that the vectors obtained as  
projections of ordered products of currents acting on the vacuum vector
 satisfy the defining relations of the off-shell Bethe vectors in a 
 generic $\fg$-invariant integrable model. 
It allows to describe a fixed normalization of the off-shell Bethe vectors 
 given by the corollary~\ref{normal} to this theorem.
 Using the definition~\ref{BV-proj} of the 
 off-shell Bethe vectors, their analytical properties are investigated in
Section~\ref{an-pro-BV}.  Appendix~\ref{AppB} contains examples of the 
off-shell Bethe vectors for the algebra $\fs\fl_3$.

\section{Two realizations of the Yangian double $\DYg$}\label{g-all}

We use the same notation as in the paper \cite{LPR-SR-CC}.
We  recall them shortly for  completeness. 
Fix a positive integer $n\geq2$. 
Let $\mathfrak{g}$ be one of the Lie 
algebras $\mathfrak{sl}_n$, $\mathfrak{o}_{2n+1}$, 
$\mathfrak{sp}_{2n}$ and $\mathfrak{o}_{2n}$ 
corresponding to the $A_{n-1}$, $B_n$, $C_n$ and $D_n$  classical series. 
Let $e_a$ be an orthonormal basis 
\begin{equation}\label{ortbas}
(e_a,e_b)=\delta_{a,b},\qquad a,b=1,\ldots,n
\end{equation}
with respect to the natural scalar product in $\CC^n$.

Let $\PSR_\ggo$ be the following set of labels of the Dynkin diagram for $\ggo$  \begin{equation}\label{PSR}
\PSR_\ggo=\begin{cases}
\{1,2,\ldots,n-1\},&\ggo=\gga\,,\\
\{0,1,2,\ldots,n-1\},&\ggo\not=\gga\,.\\
\end{cases}
\end{equation}
The cardinality of $\PSR_\fg$ is the rank of $\fg$. The positive simple roots $\rt_a$, $a\in\PSR_\ggo$ for the finite-dimensional
algebra $\ggo$ are 
\begin{equation}\label{sim-r}
\begin{aligned}
 &&&\rt_{a}=e_{a+1}-e_{a},  \quad&&a=1,\ldots,n-1,  
 \quad&&\ggo=\mathfrak{sl}_{n}\,, \\
\rt_0&=e_1, &&\rt_a=e_{a+1}-e_a, &&a=1,\ldots,n-1,  && 
\ggo=\mathfrak{o}_{2n+1}\,,\\
\rt_0&=2\,e_1, &&\rt_a=e_{a+1}-e_a,  &&a=1,\ldots,n-1, 
&&\ggo=\mathfrak{sp}_{2n}\,,\\
\rt_0&=e_2+e_1,\ &&\rt_a=e_{a+1}-e_a,  &&a=1,\ldots,n-1,   
&&\ggo=\mathfrak{o}_{2n}\,.
\end{aligned}
\end{equation}

Let $\fb_{a,b}$ and $\ago_{a,b}$ be respectively the normalized
matrix of the simple root scalar products
  and the Cartan matrix of $\ggo$:
\begin{equation}\label{Cmat}
 \fb_{a,b}=(\rt_a,\rt_b)/2=d_a\, \ago_{a,b}/2,\qquad 
d_a=(\rt_a,\rt_a)/2\,,
\end{equation}
where $(\cdot,\cdot)$ is  the scalar product \r{ortbas}. 
We introduce the rational functions $g(u,v)$, $h_a(u,v)$, $f_a(u,v)$, 
$\bgam_a(u,v)$ and $h_{a,b}(u,v)$ for $a\not=b$ as follows
\begin{subequations}\label{2.5}
\begin{equation}\label{g-h-f-gen}
g(u,v)=\frac{c}{u-v},\quad 
h_a(u,v)=\fb_{a,a}+g(u,v)^{-1},\quad 
h_{a,b}(u,v)=2\,\fb_{a,b}+g(u,v)^{-1}\,,
\end{equation}
\begin{equation}\label{g-gen}
f_a(u,v)=h_a(u,v)\,g(u,v),\quad \bgam_a(u,v)=
\frac{f_a(u,v)}{h_a(u,v)^{\zeta_a}\,h_a(v,u)^{\zeta_a}},\quad
\zeta_a=\frac{(3-\,d_{a})(2d_{a}-1)}{d_{a}+1}\,.
\end{equation}
\end{subequations}
For the simple roots $\rt_a$ such that $\fb_{a,a}=1$ we will 
use simplified notations for the functions \r{2.5} 
\begin{equation}\label{h-f-sim}
\begin{split}
h_a(u,v)\equiv h(u,v)&=\frac{u-v+c}{c},\qquad
f_a(u,v)\equiv f(u,v)=\frac{u-v+c}{u-v}\,,\\
\gamma_a(u,v)\equiv \gamma(u,v)&=\frac{f(u,v)}{h(u,v)\,h(v,u)}=
\frac{g(u,v)}{h(v,u)}=\frac{c^2}{(u-v)(v-u+c)}\,.
\end{split}
\end{equation}

In order to describe Yangian doubles simultaneously for all $\ggo$, 
we introduce two discrete parameters $\frd_\ggo$ and $\epsi_\ggo$ 
\begin{equation}\label{pri-epsi}
\frd_\ggo=\begin{cases}
n+1,&\mathfrak{g}=\mathfrak{sl}_{n}\,,\\
0,&\mathfrak{g}=\mathfrak{o}_{2n+1}\,,\\
1,&\mathfrak{g}=\mathfrak{sp}_{2n},\ \mathfrak{o}_{2n}\,,
\end{cases}\qquad 
\epsi_{\ggo}=\begin{cases}
0,& \ggo=\gga,\\
1,&\ggo=\goN,\\
-1,& \ggo=\gsp\,.
\end{cases}
\end{equation}

For any integer $i$ we define the map $i\to i^\prime$ as 
\begin{equation}\label{prime}
i^{\,\prime}=\frd_\ggo-i\,.
\end{equation}
This map acts invariantly on the set  
\begin{equation}\label{I-g}
\Ig=\{n',n'+1,\ldots,n-1,n\}\,.
\end{equation}
The index $i\in\Ig$ will be used further to numerate the matrix entries 
of a generic monodromy matrix in a generic $\ggo$-invariant 
integrable model. 
Note that the set $\Ig$ contains 
the sets $\PSR_\fg$ and $(\PSR_\fg)'$:
$\Ig=\PSR_\fg\cup(\PSR_\fg)'\cup\{n',n\}$.

 Denote by $N_\ggo$ the cardinality of the set $\Ig$. 
It is easy to verify that 
\begin{equation}\label{N-g}
N_\ggo=|\Ig|\equiv N=2n+1-\frd_\ggo\,.
\end{equation} 
Explicitly, $N_\ggo$ is equal to $n$, $2n+1$, $2n$ and $2n$
for $\ggo=\gga$, $\ggb$, 
$\gsp$, $\ggd$  respectively, and coincides with the dimension of the 
vector or first fundamental representation of $\ggo$.

For any node $a\in\PSR_\ggo$ of the Dynkin diagram we define 
two parameters $\nusp_a$ and $\nus_a$ as
\begin{equation}\label{nus-all}
\nusp_a=1+\delta_{a,0}\,(1-\epsi_\ggo)/2,\qquad
\nus_a=1+\delta_{a,0}\,(\frd_\ggo+(\epsi_\ggo-1)/2)\,.
\end{equation}
Note that for $\ggo=\gga$ we have $\nusp_a=\nus_a=1$ 
for $a\in\PSR_{\gga}$ since the
Dynkin diagram for this algebra does not contain the node with the label 0.

\subsection{Yangian double $\DYg$}\label{sect21}

For any two labels $a,b\in\PSR_\ggo$, the
rational functions   $\cre_{a,b}(u,v)$
are defined as follows
\begin{equation}\label{CuCu}
\cre_{a,b}(u,v)=\begin{cases}
h_{a,b}(u,v),&b>a,\quad \fb_{a,b}\not=0\,,\\
g(v,u)^{-1},&b<a,\quad \fb_{a,b}\not=0\,,\\
h_a(u,v),&b=a,\quad \fb_{a,b}\not=0\,,\\
1,& \fb_{a,b}=0\,.
\end{cases}
\end{equation}
 
Let $F_a[\ell]$, $E_a[\ell]$, $\hC_a[\ell]$, $\ell\in\ZZ$,  $a\in\PSR_\ggo$
be infinite set of generators.  Let $\DXg$ be the associative algebra 
built from these generators. 
Let $\mathbf{1}\in\DXg$ 
be the unity generator.  Each generator of the algebra $\DXg$ has a degree 
defined by the rules 
\begin{equation}\label{degY}
{\rm deg}\,(F_a[\ell])={\rm deg}\,(E_a[\ell])={\rm deg}\,(\hC_a[\ell])=\ell,\qquad
{\rm deg}(\mathbf{1})=0\,.
\end{equation}
In order to describe the relations between generators of the algebra $\DXg$ 
it is convenient to gather them into formal generating series 
\begin{equation}\label{gen-ser-g}
\begin{aligned}
F_a(u)&=\sum_{\ell\in\ZZ}F_a[\ell](u/c)^{-\ell-1},
\qquad &E_a(u)&=\sum_{\ell\in\ZZ}E_a[\ell](u/c)^{-\ell-1}\,,\\
\hC_a^+(u)&=\nusp_a\,\mathbf{1}+\sum_{\ell\geq 0}\hC_a[\ell](u/c)^{-\ell-1},
\qquad
&\hC_a^-(u)&=\nusp_a\,\mathbf{1}-\sum_{\ell< 0}\hC_a[\ell](u/c)^{-\ell-1}\,,
\end{aligned}
\end{equation}
using formal parameter $u$.  We will call the generating series 
\r{gen-ser-g} the {\sl currents}. 

The relations between generators of the algebra $\DXg$ are given 
by the relations between currents ($a,b\in\PSR_\fg$)
\begin{subequations}\label{CuCoRe}
\begin{equation}\label{KK}
\hC^{q_1}_a(u)\,\hC^{q_2}_b(v)=\hC^{q_2}_b(v)\,\hC^{q_1}_a(u),\quad 
q_{1,2}=\pm\,,
\end{equation}
\begin{equation}\label{KF}
\cre_{a,b}(u,v)\,\hC^{\pm}_a(u)\, F_b(v) +
\cre_{b,a}(v,u)\,F_b(v)\, \hC^{\pm}_a(u)=0\,,
\end{equation}
\begin{equation}\label{FF}
\cre_{a,b}(u,v)\, F_a(u)\,F_b(v)+  \cre_{b,a}(v,u)\, F_b(v)\,F_a(u)=0\,,
\end{equation}
\begin{equation}\label{KE}
\cre_{a,b}(u,v)\, E_b(v)\, \hC^{\pm}_a(u)+
\cre_{b,a}(v,u)\, \hC^{\pm}_a(u)\, E_b(v)=0\,,
\end{equation}
\begin{equation}\label{EE}
\cre_{a,b}(u,v)\, E_b(v)\,E_a(u)+ \cre_{b,a}(v,u)\, E_a(u)\,E_b(v)=0\,,
\end{equation}
\begin{equation}\label{EF}
\begin{split}
[E_a(u),F_b(v)]
&=c\ \delta_{a,b}\ \delta(u,v)\Big(\hC^+_{a}(u)-\hC^-_{a}(v)\Big)\,,
\end{split}
\end{equation}
\begin{equation}\label{SR-g}
\begin{split}
\mathop{\rm Sym}\limits_{u_1,\ldots,u_{m_{a,b}}}
\Big[F_a(u_1),\ldots[F_a(u_{m_{a,b}}),F_{b}(v)]\ldots\Big]&=0\,,\\
\mathop{\rm Sym}\limits_{u_1,\ldots,u_{m_{a,b}}}
\Big[E_a(u_1),\ldots[E_a(u_{m_{a,b}}),E_{b}(v)]\ldots\Big]&=0\,.
\end{split}
\end{equation}
\end{subequations}
In the relation \r{EF} $\delta(u,v)$ is  the formal series 
\begin{equation}\label{delta}
\delta(u,v)=\frac{1}{u}\sum_{\ell\in\ZZ}\frac{v^\ell}{u^\ell}\,.
\end{equation}
In \r{SR-g} $a\not=b$, $m_{a,b}=1-\ago_{a,b}$ and 
${\rm Sym}_{u_1,\ldots,u_{m_{a,b}}}$ stands for  
 the sum over all permutations $\sigma$ of the formal 
 parameters $u_1,\ldots,u_{m_{a,b}}$:
\begin{equation}\label{Sym}
\mathop{\rm Sym}\limits_{u_1,\ldots,u_{m_{a,b}}}G(u_1,\ldots,u_{m_{a,b}})=
\sum_{\sigma\in S_{m_{a,b}}}G(u_{\sigma(1)},\ldots,u_{\sigma(m_{a,b})})
\end{equation}
for any formal series $G(u_1,\ldots,u_{m_{a,b}})$ of the parameters 
$u_1,\ldots,u_{m_{a,b}}$. The relations \r{SR-g} are called the {\sl Serre relations}.

The relations between generators of the algebra $\DXg$ can be obtained from 
the relations between currents \r{CuCoRe}
by equating the coefficients at all monomials $u^{\ell_1}\,v^{\ell_2}$
for all $\ell_1,\ell_2\in\ZZ$ in left and right hand sides of the 
equalities \r{CuCoRe}. For example, the relation for the currents 
\r{FF}  is equivalent to the relation
\begin{equation}\label{cr-m5}
\Big[F_a[\ell_1+1],F_b[\ell_2]\Big]-\Big[F_a[\ell_1],F_b[\ell_2+1]\Big]=-\ 
\fb_{a,b}\begin{cases}
2\,F_a[\ell_1]\,F_b[\ell_2],& a<b\,,\\
\Big\{F_a[\ell_1],F_b[\ell_2]\Big\},&a=b\,,
\end{cases}
\end{equation}
where $\{A,B\}=A\,B+B\,A$ is anticommutator. 

\begin{rem}
Yangian double used in this paper slightly differs from the algebra defined in \cite{KhT-DY} 
in the framework of the quantum double construction of the Yangian algebra $Y(\fg)$ 
introduced in \cite{Dnew}. In \cite{LPR-SR-CC} this difference was explained by the 
shifts of the formal generating parameters used to describe the currents \r{gen-ser-g}.
Due to this shifts the Yangian double as an associative algebra 
defined by the commutation relations \r{CuCoRe}
is isomorphic to the algebra investigated in \cite{KhT-DY}. 
We used this version of Yangian double because it is more adapted for applications
in quantum integrable models governed by the $R$-matrix \r{RABCD} given below. 
\end{rem}

Algebra $\DXg$ admits a $\ZZ$-filtration 
\begin{equation}\label{filtr}
\cdots \subset \DXg_{-\ell}\subset \cdots\subset  \DXg_{-1}\subset
 \DXg_{0}\subset  \DXg_{1}\subset\cdots  \DXg_{\ell}\subset\cdots \subset \DXg
\end{equation}
defined by degrees \r{degY}. In \r{filtr} we define $\DXg_{\ell}$ as the 
linear span of the elements from $\DXg$ with degree less or equal to 
 $\ell\in\ZZ$. Let $\overline\DXg$ be the corresponding formal 
completion of $\DXg$.

One can see from the defining relations \r{CuCoRe} for the algebra $\overline\DXg$ 
  that this algebra has two distinguished associative subalgebras.  
One subalgebra is formed by the generators 
$F_a[\ell_1]$,  $\ell_1\in\ZZ$, $\hC_a[\ell_2]$, $\ell_2\geq 0$, $a\in\PSR_\ggo$
which we denote as $Y_F$. The other subalgebra is built from the generators  
$E_a[\ell_1]$,  $\ell_1\in\ZZ$, $\hC_a[\ell_2]$, $\ell_2< 0$, $a\in\PSR_\ggo$.
We denote this subalgebra as $Y_E$. 
Using approach of the paper 
\cite{KhT-DY} and coalgebraic structure \cite{D-FTINT}  of the subalgebra  $Y_F$
known as the Drinfeld's coproduct 
one can build the Yangian double $\DYg$ as the quantum double \cite{D} of this 
subalgebra and prove that as associative algebra Yangian double $\DYg$ is isomorphic 
to $\overline\DXg$.

As it was argued in \cite{LPR-SR-CC}, to be compatible with the category of highest weight representations,  the subalgebras 
 $Y_F$ and $Y_E$ have to extended to the completed subalgebra $\oUF$ and $\overline{Y}_E$.
The extended algebra\footnote{
Note that in the present paper, the definition of the completed subalgebra $\oUF$ includes  the generators 
$\hC_{a_i}[\ell_i]$, $\ell_i\geq 0$, which is not the case in the paper \cite{LPR-SR-CC}.}
  $\oUF$
is  formed by linear combinations 
of series given as infinite sums of monomials 
$y_{a_1}[\ell_1]\cdots y_{a_k}[\ell_1]$ with $\ell_1\leq \cdots\leq \ell_k$
such that $\ell_1+\cdots+\ell_k$ is fixed, where $y_{a_i}[\ell_i]$ 
is either $F_{a_i}[\ell_i]$, $\ell_i\in\ZZ$ or $\hC_{a_i}[\ell_i]$, $\ell_i\geq 0$
and $a_i\in\PSR_\fg$. The extended algebra  $\overline{Y}_E$
is defined similarly. 

One can consider the algebra generated by the elements  
$F_a[\ell]$, $E_a[\ell]$, $\hC_a[\ell]$ for $\ell\geq 0$,  $a\in\PSR_\ggo$
satisfying the same commutation and Serre  relations which can be obtained from \r{CuCoRe}
after restriction them to non-negative values of the modes index $\ell\geq 0$. 
As an associative algebra it is 
isomorphic to the Yangian $Y(\fg)$ \cite{Dnew}. Using the coalgebraic structure of the 
Yangian $Y(\fg)$ one can construct its double   $\RDYg$ which will be 
 isomorphic to $\overline\DXg$ as an algebra.

As associative algebras, the Yangians doubles 
 $\DYg$ and  $\RDYg$ are isomorphic, but as Hopf algebras 
 they are different  (see section~\ref{cc-property}). The Yangian double  $\RDYg$
  has a $RTT$ realization which will be described in the next section.

\subsection{Relation to $\ggo$-invariant integrable models}
\label{int-mod}

The Yangian double  $\RDYg$ has a realization that is more adapted 
for the applications to  $\ggo$-invariant integrable models. 
This realization  relies on the $RTT$ commutation relations with a
$\ggo$-invariant $R$-matrix. It is defined as follows.  

Let  $\epsi_i$, $i\in\Ig$ be the integer parameters 
\begin{equation}\label{eps}
\epsi_{i}=(\epsi_\ggo)^{\delta_{i>0}}\,,
\end{equation}
where $\epsi_\ggo$ is defined by \r{pri-epsi} for different $\ggo\not=\gga$.  
For any matrix $M\in{\rm End}(\CC^{N})$ 
we denote by $M^{\rm t}$ the transposition using the map \r{prime}
and defined as
\begin{equation}\label{transpo}
\begin{aligned}
&\big(M^{\rm t}\big)_{i,j}=M_{j',i'},\qquad &\ggo=\gga\,,\\
&\big(M^{\rm t}\big)_{i,j}=
\epsi_{i}\,\epsi_{j}\,M_{j',i'},\qquad &\ggo\not=\gga\,.
\end{aligned}
\end{equation}

Let $R(u,v)$ be the $\mathfrak{g}$-invariant $R$-matrix  \cite{Y,ZZ79,AMR}
\begin{equation}\label{RABCD}
R(u,v) = \begin{cases}
  \mathbf{I}\otimes\mathbf{I} + g(u,v)\, \mathbf{P},&\fg=\gga\,,\\  
   \mathbf{I}\otimes\mathbf{I} + g(u,v)\, \mathbf{P} - 
  g(u+c\,\kappa_\ggo,v)\,\mathbf{Q},&\fg\not=\gga\,,
  \end{cases}
\end{equation}
where $\mathbf{I}=\sum_{i\in\Ig}\Ee_{i,i}$ is the identity operator acting in the space $\CC^{N}$ and
$\Ee_{i,j}\in{\rm End}(\CC^{N})$ are $N\times N$ 
matrices with the only nonzero entry equals to 1 at 
the intersection of the $i$-th row and $j$-th column.

The operators $\mathbf{P}$ 
and $\mathbf{Q}$ act in 
$\CC^{N}\ot\CC^{N}$ such that
\begin{equation}\label{PQ}
\mathbf{P}=\sum_{i,j\in \Ig}\Ee_{i,j}\otimes\Ee_{j,i}, \qquad 
\mathbf{Q}
=\sum_{i,j\in \Ig}\epsi_{i}\,\epsi_{j}\,\Ee_{i,j}\otimes\Ee_{i',j'}\,.
\end{equation}
Note that for $\ggo\not=\gga$ we have that $\mathbf{Q}=\mathbf{P}^{{\rm t}_2}
=\mathbf{P}^{{\rm t}_1}$.

The parameter $\kappa_\ggo$ entering the definition of the $\ggo$-invariant 
$R$-matrix \r{RABCD}  for $\fg\not=\gga$  is equal to 
\begin{equation}\label{kappa}
\kappa_\ggo=n+\theta_\ggo,\qquad \theta_\ggo=
\epsi_\ggo+(1-\frd_\ggo)/2=
\begin{cases}
-1/2,&\mathfrak{g}=\mathfrak{o}_{2n+1},\\
1&\mathfrak{g}=\mathfrak{sp}_{2n}\,,\\
-1,&\mathfrak{g}=\mathfrak{o}_{2n}\,.
\end{cases}
\end{equation}
We do not need to define 
$\theta_\ggo$ and correspondingly 
$\kappa_\ggo$ for $\ggo=\gga$
since  the
operator $\mathbf{Q}$ is absent in $R$-matrix \r{RABCD}.

The Yangian double  $\RDYg$ is described  
by two  $N\times N$  matrices $T^\pm(u)$  satisfying 
 the commutation relations 
 \begin{equation}\label{RTT}
  R(u,v)\cdot \left( T^{q_1}(u)\otimes\mathbf{I} \right) \cdot
  \left( \mathbf{I}\otimes T^{q_2}(v) \right) =
  \left( \mathbf{I}\otimes T^{q_2}(v) \right) \cdot
  \left( T^{q_1}(u)\otimes\mathbf{I} \right) \cdot R(u,v)\,,
\end{equation}
where $q_{1,2}=\pm$. 
The  matrices $T^\pm(u)$ are Laurent series with respect 
to the formal parameter $u$ of the form 
\begin{equation}\label{T-Laut}
T^+(u)=\mathbf{I}+\sum_{m\geq 0}T^+[m]\,(u/c)^{-m-1},\qquad
T^-(u)=\mathbf{I}+\sum_{m< 0}T^-[m]\,(u/c)^{-m-1}\,.
\end{equation}
The $R$-matrix \r{RABCD} depends on formal spectral parameters 
$u$, $v$ and the rational functions of these parameters entering 
the definition of $R(u,v)$ in $\r{RTT}$ 
should be understood as Laurent series with respect 
to the ratio $(u/v)^{\pm1}$ for $q_2=-q_1=\pm$. 
When $q_1=q_2$ the rational functions  can be  Laurent series  with respect 
to either the ratio $u/v$ or the ratio $v/u$.

In the context of generic $\ggo$-integrable models 
the matrix $T^+(u)$ is identified with a universal  monodromy matrix 
for the corresponding model.

The commutation relations for the 
monodromy matrix entries $T_{i,j}^\pm(u)$ defined 
by 
\begin{equation}\label{T-ent}
T^\pm(u)=\sum_{ i,j\in \Ig} \Ee_{ij} \, T^\pm_{i,j}(u)
\end{equation}
can be written as follows 
\begin{equation}\label{rtt}
\begin{split}
&  \left[ T^{q_1}_{i,j}(u), T^{q_2}_{k,l}(v) \right] =  
  g(u,v)\,\left( T^{q_2}_{k,j}(v)T^{q_1}_{i,l}(u) - 
  T^{q_1}_{k,j}(u) T^{q_2}_{i,l}(v) \right)\\
  &\quad+g(u+\kappa_\ggo,v)\,\sum_{p\in \Ig}\epsi_{p}\left(\delta_{k,i'}\,
  \epsi_{i} \, T^{q_1}_{p,j}(u)T^{q_2}_{p',l}(v)-
  \delta_{l,j'}\,\epsi_{j}\,  T^{q_2}_{k,p'}(v)T^{q_1}_{i,p}(u)\right).
  \end{split}
\end{equation}
In the definitions \r{T-Laut} and \r{T-ent} 
 the monodromy entries $T_{i,j}(u)$ are  Laurent series of the 
formal spectral parameter $u$:
\begin{equation}\label{T-eL}
T^+_{i,j}(u)=\delta_{i,j}+\sum_{m\geq 0}T^+_{i,j}[m]\,(u/c)^{-m-1},\qquad
T^-_{i,j}(u)=\delta_{i,j}+\sum_{m< 0}T^-_{i,j}[m]\,(u/c)^{-m-1}\,.
\end{equation}

The zero mode operators $\TT_{i,j}=T^+_{i,j}[0]$ plays an important role. 
They form an algebra 
\begin{equation*}
 \left[ \TT_{i,j}, \TT_{k,l} \right] =  
  \delta_{i,l}\, \TT_{k,j}- 
  \delta_{k,j}\, \TT_{i,l} + 
  \epsi_{i}\,\epsi_{j} \left(\delta_{k,i'}\,
 \TT_{j',l}-
  \delta_{l,j'}\,  \TT_{k,i'}\right)
\end{equation*} 
which, when adding the relations following from 
\r{sbcd1} or \r{q-det}, is isomorphic to the Lie algebra $\ggo$.  
The commutation relations of $\TT_{i,j}$ with the $T$-matrices
entries $T^\pm_{k,l}(v)$ are  
\begin{equation}\label{zm-Tij}
 \left[ \TT_{i,j}, T^{\pm}_{k,l}(v) \right] =  
  \delta_{i,l}\, T^{\pm}_{k,j}(v)- 
  \delta_{k,j}\, T^{\pm}_{i,l}(v) + 
  \epsi_{i}\,\epsi_{j} \left(\delta_{k,i'}\,
 T^{\pm}_{j',l}(v)-
  \delta_{l,j'}\,  T^{\pm}_{k,i'}(v)\right)\,.
\end{equation}

Using properties of the $R$-matrix \r{RABCD} one can also verify that 
the algebra  $\RDYg$ in its $R$-matrix realization  
possesses two  automorphisms 
\begin{equation}\label{2am}
T^\pm(u)\to \Big(\big(T^\pm(u)\big)^{-1}\Big)^{\rm t}\quad\mbox{and}\quad 
T^\pm(u)\to \Big(\big(T^\pm(u)\big)^{\rm t}\Big)^{-1}\,.
\end{equation}

For $\ggo=\goN$ and $\gsp$ there are
symmetry relations 
\begin{equation}\label{sbcd}
T^\pm(u-c\,\kappa_\ggo)^{\rm t}\cdot T^\pm(u)= 
\zcent^\pm(u)\ \mathbf{I},
\end{equation}
in the Yangian doubles  $\RDYg$ and the
two automorphisms \r{2am} differ only by a shift of the spectral parameter
and central elements $\zcent^\pm(u)$. 
In what follows we will fix these central elements  by the relation 
$\zcent^\pm(u)=1$ so that
the symmetry relations  simplify to 
\begin{equation}\label{sbcd1}
\Big(\big(T^\pm(u)\big)^{\rm t}\Big)^{-1}=T^\pm(u+c\,\kappa_\ggo)
\quad\mbox{or}\quad 
\Big(\big(T^\pm(u)\big)^{-1}\Big)^{\rm t}=T^\pm(u-c\,\kappa_\ggo)\,.
\end{equation}

\subsection{Ding-Frenkel isomorphism  between $\DYg$ and  $\RDYg$}

The isomorphism of the Yangian doubles 
 $\DYg$ and  $\RDYg$ as associative algebras can be 
observed in the framework of the Ding-Frenkel approach introduced 
in \cite{DF} for the quantum affine algebra $U_q(\gln)$. 
For the Yangian $Y(\ggo)$ this approach was developed in \cite{JLM18}
and for Yangian doubles $\DYg$ and  $\RDYg$ in \cite{LP1}. 

The starting point of the Ding-Frenkel construction is 
the Gaussian decomposition of the  $T$-operators 
for the Yangian double  $\RDYg$
associated with the vector representation of 
$\ggo$
\begin{equation}\label{Gauss}
T^\pm_{i,j}(u)=\sum_{\ell\,\in\,\Ig}
\FF^\pm_{\ell,i}(u)\,k^\pm_\ell(u)\,\EE^\pm_{j,\ell}(u)
\end{equation}
 such that 
\begin{equation*}
\FF^\pm_{i,i}(u)=\EE^\pm_{i,i}(u)=1,\qquad 
\FF^\pm_{j,i}(u)=\EE^\pm_{i,j}(u)=0,\qquad j<i,\qquad i,j\in \Ig\,,
\end{equation*}
where $\Ig$ is the set of indices defined by \r{I-g}.

Using the parameters \r{nus-all}
we can express the simple roots currents 
$F_a(u)$, $E_a(u)$ and $\hC^\pm_a(u)$ through
the Gaussian coordinates using the Ding-Frenkel type formulas \cite{BK, DF} 
\begin{equation}\label{sr-cur}
\begin{split}
F_a(u)&=\FF^+_{a+\nus_a,a}(u)-\FF^-_{a+\nus_a,a}(u)\,,\\
E_a(u)&=\EE^+_{a,a+\nus_a}(u)-\EE^-_{a,a+\nus_a}(u)\,,\\
\hC^\pm_a(u)&=\nusp_a\,k^\pm_a(u)\,k^\pm_{a+\nus_a}(u)^{-1}\,.
\end{split}
\end{equation}

The diagonal Gaussian coordinates $k^\pm_a(u)$  are algebraically dependent according 
to the relation 
\begin{equation}\label{q-det}
\mbox{qdet}\Big(T^\pm(u)\Big)=
\sum_{\sigma\in S_{n}}(-1)^{|\sigma|}\ T^\pm_{1,\sigma(1)}(u)\,
T^\pm_{2,\sigma(2)}(u-c)\cdots T^\pm_{n,\sigma(n)}(u-c(n-1))=1
\end{equation}
for $\ggo=\gga$ and to the relations \r{sbcd1} for $\ggo=\goN$, $\gsp$. 
Using the Gaussian decomposition \r{Gauss} one can show that 
these relations are equivalent to 
\begin{equation}\label{sln-cond}
k^\pm_{1}(u)\,k^\pm_2(u-c)\cdots k^\pm_{n-1}(u-c(n-2))\,k^\pm_{n}(u-c(n-1))=1
\end{equation}
for $\ggo=\gga$ and for $\ggo\not=\gga$ to the relations 
\begin{equation}\label{cond-g}
\begin{aligned}
k^\pm_0(u+c/2)\,k^\pm_0(u)&=
\prod_{s=1}^n k_s^\pm(u-c(s-3/2))\,k_s^\pm(u-c(s-1/2))^{-1},\quad
&&\ggo=\ggb\,,\\
k^\pm_0(u)\,k^\pm_1(u-2\,c)&=
\prod_{s=2}^n k_s^\pm(u-c\,s)\,k_s^\pm(u-c(s+1))^{-1},\quad
&&\ggo=\gsp\,,\\
k^\pm_0(u)\,k^\pm_1(u)&=
\prod_{s=2}^n k_s^\pm(u-c\,(s-2))\,k_s^\pm(u-c(s-1))^{-1},\quad
&&\ggo=\ggd\,.
\end{aligned}
\end{equation}

When the relation 
\r{sln-cond} is not imposed while all other relations remain, one gets
 the Yangian double  $\RDYAn$. On the other hand 
 the Yangian double  $\RDYSn$ can be identified with the quotient of the Yangian double 
  $\RDYAn$ by the relation \r{q-det} in the $RTT$ formulation or by 
 the relation \r{sln-cond} in the Drinfeld's 'new' realization. 

The Ding-Frenkel formulas \r{sr-cur} display a map from the generators of 
the Yangian double  $\RDYg$ in its $R$-matrix realization to the 
Yangian double $\DYg$ in its 'new' realization \cite{Dnew}. To describe the 
inverse map we need to introduce several elements related to the projection 
method. The full description of this method for quantum affine algebras 
can be found in \cite{EKhP}. For the Yangian double $\DYg$ a partial description
can be found in \cite{HLPRS17,LP1}.

\subsection{Elements of the projection method}\label{cc-property}

Let  $\overline{Y}_F$ be the completed subalgebra in the Yangian double 
$\DYg$ generated by the currents $F_a(u)$ and $\hC^+_a(u)$
as it was defined in the section~\ref{sect21}. 
Recall that we consider any product of  simple root currents 
in the category of  highest weight representations which means 
that it belongs to the 
completed subalgebra $\overline{Y}_F$ in the Yangian double 
$\DYg$ and is normal ordered. The normal ordering of product of 
 currents is a presentation of the product 
$F_{a_1}(z_1)\cdots F_{a_m}(z_m)$ in the form 
when all currents $F_{a_s}(z_s)$, $s=1,\ldots,m$ are 
replaced by the difference of the Gaussian coordinates 
$\FF^+_{a_s+\nus_{a_s},a_s}(z_s)-\FF^-_{a_s+\nus_{a_s},a_s}(z_s)$ according 
to the Ding-Frenkel formulas \r{sr-cur} and then all 'negative' 
Gaussian coordinates are moved to the left of all 'positive' Gaussian coordinate using the commutation relations \r{RTT}
in the  Yangian double $\RDYg$. 

As it is seen from the $RTT$ commutation relations \r{RTT},  
the Yangian double $\RDYg$  has two 
 Borel subalgebras 
$\sB^\pm$ generated by the modes of the  $T$-operators 
$T^\pm(u)$. It was already mentioned that the subalgebra $\sB^+$ 
is isomorphic to the Yangian $Y(\fg)$. 
Alternatively, the  Borel subalgebras 
$\sB^\pm$ can be seen as generated by the algebraically dependent 
Gaussian coordinates $\FF^\pm_{j,i}(u)$, 
$\EE^\pm_{i,j}(u)$, and $k^\pm_j(u)$ for $i<j$ and $i,j\in\Ig$. 
The equivalence between these two descriptions 
follows from the invertibility of the  
Gaussian decomposition  and the fact that 
the Gaussian coordinates can be uniquely expressed through 
quasi-minors \cite{GeRe} of the $T$-operators $T^\pm(u)$ \cite{BK,M}.

The Yangian doubles $\DYg$ and $\RDYg$ are isomorphic as 
associative algebras and the Ding-Frenkel formulas \r{sr-cur} 
describe this isomorphism.  
Due to this isomorphism, the
subalgebras $\sB^\pm\in\RDYg$ and $\oUF\in\DYg$ 
have non-empty intersections
\begin{equation}\label{inter}
U^\pm_F=\sB^\pm\cap  \oUF\,.
\end{equation}
It follows from the definition of the Borel subalgebras that 
one can consider the intersection $U^-_F$ as the linear span 
of monomials composed from the Gaussian coordinates $\FF^-_{j,i}(u)$. They are algebraically dependent due to 
the symmetry relations \r{sbcd1}.
Similarly, the intersection $U^+_F$ is viewed as the linear span 
of monomials composed from the 
Gaussian coordinates $\FF^+_{j,i}(u)$
and  $k^+_{j}(u)$. Here the Gaussian coordinates 
$\FF^\pm_{j,i}(u)$, $i<j$, $i,j\in\Ig$ can be related 
to the simple root currents $F_a(u)$, $a\in\PSR_\fg$ by the projections \r{proj} 
on the intersections \r{inter} 
 and acting on the composed currents. The composed 
currents are defined as  products 
of  simple root currents and are given by the formulas 
\r{ccABC} and \r{ccD} below. It was proved in \cite{LPR-SR-CC} that they 
are well defined elements of the completed subalgebra  $\oUF$ and their 
properties can be inherited from the Serre relations.

The commutation relations \r{CuCoRe} 
in the Drinfeld's realization of $\DYg$ or the commutation relation 
\r{RTT} in the $R$-matrix realization of  $\RDYg$ allows 
to present an arbitrary element  $\cF\in \oUF$ in the normal ordered form 
\begin{equation}\label{nor-ord}
\cF=\cF^-\,\cF^+,\qquad \cF^\pm\in U^\pm_F\,.
\end{equation}

The subalgebras $\sB^\pm$  
are Hopf subalgebras in the Yangian double  $\RDYg$ equipped with the coproduct
\begin{equation}\label{st-cop}
\Delta\Big(T^\pm_{i,j}(u)\Big)=\sum_{\ell\in\Ig}
T^\pm_{\ell,j}(u)\ot T^\pm_{i,\ell}(u),\qquad \varepsilon(T^\pm_{i,j}(u))=\delta_{i,j}\,,
\end{equation}
 where $\coun$ is the counit map.   
 The subalgebra  $\oUF$
 is a Hopf subalgebra in the Yangian double  $\DYg$
  equipped with Drinfeld's coproduct \cite{D-FTINT} 
\begin{equation}\label{D-cop}
\begin{aligned}
\DeltaD\Big(F_a(u)\Big)&=1\ot F_a(u)+F_a(u)\ot\hC^+_a(u),\qquad 
&\counD(F_a(u))&=0\,,\\
\DeltaD\Big(\hC^+_a(u)\Big)&=\nusp_a^{-1}\,\hC^+_a(u)\ot\hC^+_a(u),\qquad
&\counD(\hC^+_a(u))&=\nusp_a\,.
\end{aligned}
\end{equation}

In analogy with 
the general theory of the projection method 
presented in \cite{EKhP} for the quantum affine algebras, one can define 
projections onto the intersections \r{inter} for any element  $\cF\in \oUF$
written in the normal ordered form $\cF=\cF^-\,\cF^+$ as follows
\begin{equation}\label{proj}
\Pfp\Big(\cF^-\,\cF^+\Big)=\counD(\cF^-)\,\cF^+,\qquad 
\Pfm\Big(\cF^-\,\cF^+\Big)=\counD(\cF^+)\,\cF^-\,.
\end{equation}

The  projections \r{proj} yield a power tool 
to investigate the algebraic nested Bethe ansatz. For example, in \cite{KhP},
 exact formulas for off-shell Bethe vectors were calculated 
using these projections in a generic integrable model related to the quantum affine 
algebra $U_q(\widehat{\mathfrak{gl}}_n)$.
One of the main property of these projections is the possibility 
to write explicitly an arbitrary element  $\cF\in \oUF$ in its normal ordered 
form using the Drinfeld's comultiplication map  \r{D-cop}
\begin{equation}\label{nor-form}
\cF=\sum \Pfm\Big(\cF^{(2)}\Big)\,  \Pfp\Big(\cF^{(1)}\Big),\qquad
\DeltaD(\cF)=\sum \cF^{(1)}\otimes \cF^{(2)}\,.
\end{equation}

For  $\ggo\not=\gga$, let $\bvphi_\ggo$ be the parameter 
\begin{equation}\label{vphi}
\bvphi_{\mathfrak{g}}=2\,\xi_\ggo+(\epsilon_\ggo-1)/2=
\begin{cases}
0,&\ggo=\ggb\,,\\
1,&\ggo=\gsp\,,\\
2,&\ggo=\ggd\,.
\end{cases}
\end{equation}
For $\ggo\not=\gga$ and 
$s\in\PSR_\ggo$
let $z_s$ be the shifted parameters
\begin{equation}\label{z-sh1}
z_s=z-c(s+\theta_\ggo),\qquad \bvphi_\ggo\leq s\leq n-1\,,
\end{equation} 
where $\theta_\ggo$ is defined in \eqref{kappa}.

Besides the simple root currents $F_a(z)$ for $\bvphi_\ggo\leq a\leq n-1$ 
(for $\fg\not=\gga$)
given by the Ding-Frenkel formulas \r{sr-cur}, we can introduce {\sl auxiliary} 
simple root currents  \cite{LP1}  by the equality 
\begin{equation}\label{dep-cur}
F_{a'-1}(z)=-\ F_a(z_a),\qquad \bvphi_\ggo\leq a\leq n-1\,.
\end{equation}

Combining simple roots and (for $\ggo\not=\gga$) auxiliary simple roots currents 
we can define 
the {\sl composed} currents $F_{j,i}(u)$ as elements of the completed 
subalgebra  $\oUF\subset \DYg$. 
For the algebras $\ggo=\gga$, $\ggb$, and $\ggc$ they read
\begin{equation}\label{ccABC}
F_{j,i}(u)=F_i(u)\, F_{i+1}(u)\cdots F_{j-2}(u)\, F_{j-1}(u)\,,
\qquad n'\leq i<j\leq n
\end{equation} 
while for the algebra $\ggo=\ggd$ the composed currents can be presented as 

\begin{equation}\label{ccD}
F_{j,i}(u)=\begin{cases}
F_i(u)\, F_{i+1}(u)\cdots  F_{j-1}(u),\quad 2\leq i<j\leq n
\quad \mbox{or} 
\quad
2\leq j'<i'\leq n\,,\\
-\ F_{i}(u)\cdots F_{-2}(u)\, F_{j'}(u),\quad 2\leq i'\leq n,\quad j=0,1\,,\\
0,\quad i'=j=1\,,\\
 F_{i}(u)\, F_{2}(u)\cdots F_{j-1}(u),\quad i=0,1,\quad 2\leq j\leq n\,,\\
-\ \Big(F_i(u)\cdots F_{-2}(u)\Big)\, F_0(u)\, F_1(u)\, \Big(F_2(u)\cdots F_{j-1}(u)\Big),\quad
2\leq i',j\leq n\,.
\end{cases}
\end{equation}
According to the last line in \r{ccD} for 
$i'=j=2$, the composed current $F_{2,-1}(u)=-F_0(u)\,F_1(u)$, where 
$F_0(u)$ and $F_1(u)$ are commuting currents. 

Composed currents were introduced in \cite{DKh} for  quantum affine 
algebras where their correctness within the category 
of  highest degree representations was established. 
The following proposition was proved in \cite{LP1}.

\begin{prop}\label{inv-DF}
\cite{LP1} 
The Gaussian coordinates $\FF^\pm_{j,i}(u)$, $n'\leq i< j\leq n$ for the 
$T$-operator of the Yangian double $\DYg$ 
associated to the vector 
(first fundamental) representation of $\ggo$ are related to the composed 
currents $F_{j,i}(u)$ via the projections \r{proj}:
\begin{equation}\label{DF-in-m}
\FF^+_{j,i}(u)=\Pfp\Big(F_{j,i}(u)\Big),\qquad 
\tFF^-_{j,i}(u)=\Pfm\Big(F_{j,i}(u)\Big)\,,
\end{equation} 
where 
\begin{equation}\label{inv-F}
\tFF^-_{j,i}(u)=\sum_{\ell=0}^{j-i-1}(-1)^{\ell+1}
\sum_{j>i_\ell>\cdots>i_1>i}
\FF^-_{i_1,i}(u)\,\FF^-_{i_2,i_1}(u)\cdots \FF^-_{i_\ell,i_{\ell-1}}(u)\,
\FF^-_{j,i_{\ell}}(u)\,.
\end{equation}
\end{prop}

Similar proposition can be formulated for the relation between Gaussian 
coordinates $\EE^\pm_{i,j}(u)$ and the composed currents $E_{i,j}(u)$, $i<j$,
$i,j\in\Ig$. We do not need this relation in this paper.

\section{Off-shell Bethe vectors}\label{BVsect}

We wish to investigate the properties of the
Bethe vectors in a generic $\ggo$-invariant integrable model studied within the framework of the universal nested algebraic Bethe ansatz. 
For such a purpose, we define
the universal monodromy matrix of the model as the $
T$-operator $T(u)\equiv T^+(u)$
 of the Yangian double  $\RDYg$. 
 This is possible because the monodromy matrix in any 
 $\ggo$-invariant integrable model satisfies the same commutation relations 
 \r{RTT} as $T^+(u)$.

We consider generalized models, defined by a vacuum $|0\rangle$, a set of  functional parameters $\lambda_i(u)$ and the relations
\begin{equation}\label{rvec}
\Big(T_{i,j}(u)-\lambda_i(u)\,\delta_{i,j}\Big)\rvac =0\,,\quad i\geq j\,,
\quad i,j\in\Ig
\end{equation}
valid for any formal parameter $u$. It is known that all periodic spin chain models fall in this class of models. The space of 
 states in these models is  the set $\cV$ of polynomials in $T_{i,j}(t_k)$, $i<j$ acting on the vacuum $\rvac$
  \begin{equation}\label{bvgr}
\BB(\bt)=\mathcal{P}(T_{i< j}(\bt))\ \rvac\,.
\end{equation}
In \eqref{bvgr}, $\bt=\{\bt^s\}$,
$s\in\PSR_\ggo$, where $\bt^s=(t^s_1,\ldots,t^s_{|\bt^s|})$ are sets of  Bethe parameters, with cardinality 
$|\bt^s|\geq 0$, and 
$s$ is called the color of the Bethe parameters.
The coefficients of these polynomials are rational functions of the 
Bethe parameters 
and of the functions $\lambda_i(z)$.  
\begin{Def}\label{def:offBV}
    We call {\sl off-shell Bethe vectors}, the vectors 
\r{bvgr} when they are symmetric 
with respect to any permutation of the Bethe parameters of the 
same color and if they satisfy the conditions 
 formulated below in propositions~\ref{BVdef3} and \ref{BVdef1}.
 It was shown in \cite{LPR-RR} that, for a given choice of Bethe parameters $\bar t$,
 these vectors are unique. 
\end{Def}

The functions $\lambda_i(u)$ for $n'\leq i\leq n$ are functional parameters of a generic $\ggo$-invariant 
integrable model and they are 
algebraically  dependent due to the relations (the shifted parameters $z_s$
are defined in \r{z-sh1})
\begin{equation}\label{lam-con-g}
\begin{split}
\lambda_1(z)\,\lambda_2(z-c)\cdots 
\lambda_{n}(z-c(n-1))&=1,
\qquad \ggo=\gga\,,\\
\lambda_{a'}(z)\,\lambda_a(z_a)
\prod_{s=a+1}^n \frac{\lambda_s(z_{s})}{\lambda_s(z_{s-1})}&=1,\qquad
\ggo\not=\gga\,,\qquad a\in\PSR_\ggo\cup\{n\}\,.
\end{split}
\end{equation}

In what follows the scalar functions $f(\bu,\bv)$ 
or the commuting operators $\mathcal{O}(\bu)$ depending on the sets 
$\bu=\{u_1,\ldots,u_{|\bu|}\}$ and 
$\bv=\{v_1,\ldots,v_{|\bv|}\}$ will be always understood as the products 
\begin{equation}\label{conv}
f(\bu,\bv)=\prod_{u_i\in\bu}\prod_{v_j\in\bv} f(u_i,v_j),\qquad 
\mathcal{O}(\bu)=\prod_{u_i\in\bu}\mathcal{O}(u_i)\,.
\end{equation}
We set such products equal to 1 if any of the sets $\bu$ and/or $\bv$ is empty.

In the framework of the projection method, the off-shell Bethe vectors $\BB(\bt)$  can be constructed from the generators of the Yangian double $\DYg$ in its Drinfeld's 'new' realization
using projection $\Pfp$ defined by \r{proj}.

If the Yangian $Y(\fg)$ is realized by the Gaussian coordinates \r{Gauss}
of the monodromy matrix $T(u)$, then one can translate the properties 
of the vacuum vector $\rvac$ \r{rvec} to the properties 
\begin{equation}\label{rvec1}
\EE^+_{i,j}(u)\rvac=0,\quad k^+_j(u)\rvac=\lambda_j(u)\rvac,\quad i,j\in\Ig\,.
\end{equation}
The space of states $\cV$ of $\fg$-invariant integrable model
will be the set of vectors $\Bf(\bt)$ defined by polynomials 
in $\FF^+_{j,i}(t^s_k)$ acting on the vacuum vector  
\begin{equation}\label{bvgrF}
\Bf(\bt)=\mathcal{P}'(\FF^+_{j,i}(\bt))\rvac\,.
\end{equation}
To find the relation between the polynomials $\mathcal{P}'$ in \r{rvec1} and 
$\mathcal{P}$ in \r{bvgr}
one has first to  substitute Gaussian decomposition \eqref{Gauss} of monodromy entries 
$T_{i,j}(t^s_k)$ in definition \r{bvgr}. Then, one has to use the commutation relations
between Gaussian coordinates to move all $\EE^+_{i,j}(t^s_k)$ and $k^+_j(t^s_k)$ 
to the right in each monomial, and finally to apply 
the property \r{rvec1} of the vacuum vector $\rvac$. 
Let us consider an example  given by the following definition.

\begin{Def}\label{BV-proj}
Let $\Bf(\bt)\in \cV$ be the vector defined by 
\begin{equation}\label{BVg}
\Bf(\bt)=
\left(\prod_{s\in\PSR_{\ggo}}
\frac{1}{g(\bar t^{\,s+\nu_s},\bar t^{\,s})}
\prod_{\ell_1>\ell_2}^{|\bt^s|}\bgam_s(t^s_{\ell_1},t^s_{\ell_2})\right)\,
\Pfp\left(\prod_{s\in\PSR_{\ggo}}^{\longleftarrow}
\prod_{1\leq \ell\leq |\bt^s|} ^{\longleftarrow}F_s(t^s_\ell)
\right)\rvac\,,
\end{equation}
where the function $\bgam_s(u,v)$ is defined by \r{g-gen} and  
the ordered product of  non-commutative entries $A_s$
is defined as follows
\begin{equation*}
\prod_{s\in \PSR_\fg}^{\longleftarrow}A_s=
A_{n-1}\, A_{n-2}\cdots A_2\, A_1\,A_0\,.
\end{equation*}
$\Bf(\bt)$ is a polynomial in the Gaussian coordinates 
$\FF^+_{j,i}(t^s_\ell)$ and has the form given in \r{bvgrF}.
 \end{Def}
 
 The fact that projection of the product of the currents 
 produces a polynomial in the Gaussian coordinates $\FF^+_{j,i}(t^s_\ell)$ 
 can be verified by the explicit calculation of this projection (see the technique of such calculations in \cite{HLPRS17}
 for the case of supersymmetric algebra $\fg=\fg\fl(m|n)$).
 The main result of this paper is following theorem 
 \begin{thm}\label{main-thm}
  The vector $\Bf(\bt)$ introduced in \r{BVg}   
 coincides with the off-shell Bethe vector $\BB(\bt)$ 
 given by definition~\ref{def:offBV}:
 \begin{equation}\label{main-st}
 \Bf(\bt)=\BB(\bt)\,.
 \end{equation}
 \end{thm}
 \proof
Note, first of all, that according to the commutation relation \r{FF}, the  vector
$\Bf(\bt)$ is a  vector-valued function, symmetric
in the Bethe parameters $t^s_\ell$ of the same color.  According to the propositions~\ref{BVdef3}
and \ref{BVdef1} formulated below, the vector $\Bf(\bt)$ satisfies 
the defining relations  for the off-shell Bethe vectors in a generic 
$\fg$-invariant integrable model \cite{LPR-RR}. This implies that $\Bf(\bt)$ 
 satisfy the same rectangular recurrence relations as the off-shell 
Bethe vector $\BB(\bt)$ \cite{LPR-RR}. Hence 
they  coincide since $\Bf(\vn)=\BB(\vn)=\rvac$.  
Then, the theorem reduces to the proof of the propositions 
~\ref{BVdef3} and \ref{BVdef1}, done below. \qed

\begin{cor}\label{normal}
The normalization of the off-shell Bethe vectors 
given by the definition~\ref{def:offBV} is defined by 
the normalization of the main term of this vector
\begin{equation}\label{BV-mt}
\BB(\bt)=
\prod_{s\in\PSR_{\ggo}}\left(
\frac{1}{g(\bar t^{\,s+\nu_s},\bar t^{\,s})}\,
\frac{\nusp_s^{|\bt^s|}}{h_s(\bt^s,\bt^s)^{\zeta_s}}\right)\,
\prod_{s\in\PSR_{\ggo}}^{\longleftarrow}
\frac{T_{s,s+\nu_s}(\bt^s)}{\lambda_{s+\nus_s}(\bt^s)}\,\rvac+\cdots\,,
\end{equation}
where $\cdots$ stands for terms which contain the monodromy 
entries $T_{i,j}(t^s_\ell)$ for $i\leq s$, $s+\nus_s\leq j$ and $j-i>\nus_s$.
\end{cor}

The parameters $\zeta_s$, $\nusp_s$, and $\nus_s$ in \r{BV-mt}
are given by \r{g-gen} and  \r{nus-all}. 
Recall that these parameters in most of the cases are 
equal to 1 except the cases when $\zeta_0=0$ for $\fg=\fo_{2n+1}$, $\nusp_0=2$ for $\fg=\fs\fp_{2n}$,
and $\nus_0=2$ for $\fg=\fo_{2n}$.

One can note that the main term of the off-shell Bethe 
vector $\BB(\bt)$ shown in \r{BV-mt} is obviously 
symmetric with respect to permutation of the Bethe parameters 
of the same color. 
The normalization provided by this corollary is fixed according to the ordering of the noncommutative monodromy matrix entries appearing in the main term of the off-shell Bethe vector, as presented in \r{BV-mt}. This normalization factor will be changed 
accordingly, 
if one changes the order in this product of monodromy entries using the 
commutation relations \r{rtt}. 

Note also that this normalization is essentially universal, in the sense that the only model-dependent contribution is given by the product of the factors $\lambda_{s+\nus_s}(\bt^s)$.

\subsection{The highest monodromy matrix entry action}

For $\ggo\not=\mathfrak{sl}_n$ and 
$s\in\PSR_\ggo$
let  $\bz_s$ be the sets of  cardinality $1+\delta_{s\geq\bvphi_\ggo}$
defined by the equalities 
\begin{equation}\label{z-sh2}
\bz_s=\begin{cases}
\{z\}, & 0\leq s<\bvphi_\ggo\,,\\
\{z,z_s\},&\bvphi_\ggo\leq s\leq n-1\,,
\end{cases}
\end{equation}
where $z_s$ are given by
\eqref{z-sh1}.
We define the extended sets of  Bethe parameters 
\begin{equation}\label{exBCD}
\bw^s=\begin{cases}
\{\bt^s,z\},&\fg=\gga\,,\\
\{\bt^s,\bz_s\},&\fg\not=\gga\,.
\end{cases}
\end{equation}

\begin{prop}\label{BVdef3}
The action of $T_{n',n}(z)$ on the  vectors $\Bf(\bt)$ 
\r{BVg}  reads
\begin{equation}\label{T-act-g}
T_{n',n}(z)\, \Bf(\bt)=\lambda_n(z)\,\mus^n_{n'}(z;\bt)\, \Bf(\bw)\,,
\end{equation}
where the extended sets of Bethe 
parameters are given by \r{exBCD} and 
\begin{equation}\label{mus-all}
\mus^n_{n'}(z;\bt)=  \begin{cases}
h(\bt^1,z)\,h(z,\bt^{n-1}),& \ggo=\gga\,,\\[4mm]
\displaystyle \ovs_\ggo(z;\bt)\ 
\frac{h(z,\bt^{n-1})}{g(z_n,\bt^{n-1})},&\ggo\not=\gga\,,
\end{cases}
\end{equation}
with
\begin{equation}\label{ovsBCD}
\ovs_\ggo(z;\bt)= -\  \kappa_\ggo \begin{cases}
\displaystyle  
\frac{g(z-c/2,\bt^0)}{h(z,\bt^0)},&\ggo=\ggb\,,\\[4mm]
\displaystyle  
(-1)^{|\bt^0|}\,h(\bt^1,z),&\ggo=\gsp\,,\\[4mm]
\displaystyle 
(-1)^{|\bt^0|+|\bt^1|},&\ggo=\ggd\,.
\end{cases}
\end{equation}
\end{prop}
\proof 
According to the Gaussian decomposition \r{Gauss}
the highest monodromy entry has a factorized form 
\begin{equation}\label{hent}
T^+_{n',n}(z)=\FF^+_{n,n'}(z)\,k^+_{n}(z)\,.
\end{equation}
Due to the proposition~\ref{inv-DF}
the highest Gaussian coordinate $\FF^+_{n,n'}(z)$ can be expressed 
through the simple root currents $F_a(z)$, $a\in\PSR_\ggo$ and the auxiliary 
currents $F_{a'-1}(z)$, $\varphi_\ggo\leq a\leq n-1$
 \r{dep-cur} by the projection $\Pfp$ \r{proj} as follows
\begin{equation}\label{CC3}
\begin{aligned}
\FF^+_{n,n'}(z)&= \Pfp\left(
\prod^{\longrightarrow}_{s\in\PSR_\ggo}F_s(z)\right),\qquad &&\mathfrak{g}=\gga\,,\\
\FF^+_{n,n'}(z)&= (-1)^{n-\frd_\ggo}\Pfp\left(
\prod^{\longleftarrow}_{s\in\PSR_\ggo,\ s\geq \varphi_\ggo}F_s(z_s)
\ \prod^{\longrightarrow}_{a\in\PSR_\ggo}F_{a}(z)\right),\qquad &&\mathfrak{g}\not=\gga\,.
\end{aligned}
\end{equation}
Recall that $z_s=z-c(s+\theta_\ggo)$ where parameters $\theta_\ggo$ 
are equal to $-1/2,1,-1$ for $\ggb$, $\gsp$, $\ggd$ respectively \r{kappa}. 

Let $\IF$ be the left ideal of $\DYg$ 
\begin{equation}\label{Iid}
\IF=\Big(U^-_F\cap {\rm Ker}(\counD)\Big) Y(\ggo)
.
\end{equation}
The ideal $\IF$ is annihilated by the projection 
$\Pfp(\IF)=0$. Recall that $U^-_F$ is the subalgebra of the 
Yangian double $\DYg$ generated by all Gaussian coordinates $\FF^-_{j,i}(z)$
such that $i,j\in\Ig$ and $i<j$. 

The normal ordering relation \r{nor-form} together with the Drinfeld coproduct 
\r{D-cop} for the simple root currents implies that any element $\cF$ of 
the  ideal $ \oUF\cap {\rm Ker}(\counD)$ in the 
completed subalgebra  $\oUF$ can be presented in the form 
\begin{equation}\label{pap1}
\cF=\Pfp(\cF)+\mathfrak{F}\,\qquad \mathfrak{F}\in\F\,.
\end{equation}
To calculate the action of the highest monodromy $T^+_{n',n}(z)$ on 
the off-shell Bethe vector $\BB(\bt)$ \r{BVg} we need following lemma.

\begin{lemma}\label{Tac-van}\cite{LP1}
The action of the highest monodromy entry $T^+_{n',n}(z)$ on any 
element of the left ideal $\F$ belongs to the same ideal. 
\end{lemma}
\proof Using the explicit expression of the highest monodromy element 
 $T_{n',n}(z)$ \r{hent} we can prove this lemma using the
 $RTT$ commutation relation \r{rtt}. It is sufficient to prove that 
 the commutator $[T^+_{n',n}(z),\FF^-_{j,i}(u)]\in\IF$ and we refer to appendix~A in  \cite{LP1} for a detailed proof.\qed 

Let 
\begin{equation}\label{Ng}
\mathcal{N}_\ggo(\bt)=
\prod_{s\in\PSR_{\ggo}}
\frac{1}{g(\bar t^{\,s+\nu_s},\bar t^{\,s})}
\prod_{\ell_1>\ell_2}^{|\bt^s|}\bgam_s(t^s_{\ell_1},t^s_{\ell_2})
\end{equation}
be the normalization factor in the off-shell Bethe vectors 
\r{BVg} and 
\begin{equation}\label{BFt}
\BF(\bt)=\prod_{s\in\PSR_{\ggo}}^{\longleftarrow}
\prod_{1\leq \ell\leq |\bt^s|} ^{\longleftarrow}F_s(t^s_\ell)\in  \oUF
\end{equation}
be an ordered product of the simple root currents 
depending on the sets of Bethe parameters $\bt^s=\{\bt^s_1,\ldots,\bt^s_{|\bt^s|}\}$. 
Then the  vector $\Bf(\bt)$ can be written shortly as follows
\begin{equation}\label{BVsh}
\Bf(\bt)=\mathcal{N}_\ggo(\bt)\,\Pfp\Big(\BF(\bt)\Big)\rvac\,.
\end{equation}

Let $\Tb(\{y^a\},\{x^s\})$ be an element of the
completed subalgebra   $\oUF$ of the form 
\begin{equation}\label{hent1}
\Tb(\{y^a\},\{x^s\})=
F_{n,n'}(\{y^a\},\{x^s\})\,k^+_{n}(y^{n-1})\,,
\end{equation}
depending on the parameters $\{y^a\}$, $\{x^s\}$,
$a,s\in\PSR_\ggo$ and $\varphi_\ggo\leq s\leq n-1$,
where 
\begin{equation}\label{CC31}
F_{n,n'}(\{y^a\},\{x^s\})= 
\begin{cases}\displaystyle \
\prod^{\longrightarrow}_{a\in\PSR_\ggo}F_a(y^a), &\mathfrak{g}=\gga\,,\\[5mm]
\displaystyle (-1)^{n-\frd_\ggo}
\prod^{\longleftarrow}_{s\in\PSR_\ggo,\ s\geq \varphi_\ggo}
F_s(x^s-c(s+\theta_\ggo))\prod^{\longrightarrow}_{a\in\PSR_\ggo}F_{a}(y^{a}),&
\mathfrak{g}\not=\gga\,.
\end{cases}
\end{equation}

\bigskip

Note that if $y^a=x^s=z$ for all values of the indices $a,s\in\PSR_\ggo$ and 
$\varphi_\ggo\leq s\leq n-1$,  the projection of the 
element $\Tb(\{y^a\},\{x^s\})$ \r{hent1}
coincides with the highest monodromy entry $T^+_{n',n}(z)$ \r{hent}
\begin{equation}\label{hent2}
T^+_{n',n}(z)=\Pfp\Big(\Tb(\{y^a\},\{x^s\})\Big)\Big|_{y^a=x^s=z}\,.
\end{equation}

Using lemma~\ref{Tac-van}, properties of the projections  and relation \r{hent2}
we can reduce the action of the highest monodromy on the 
off-shell Bethe vector \r{BVsh} 
\begin{equation}\label{red-ac}
\begin{split}
T^+_{n',n}(z)\,\Bf(\bt)&=\mathcal{N}_\ggo(\bt)\,
T^+_{n',n}(z)\,\Pfp\Big(\BF(\bt)\Big)\rvac=\\
&=\mathcal{N}_\ggo(\bt)\,\Pfp\Big(T^+_{n',n}(z)\,\BF(\bt)\Big)\rvac=
\mathcal{N}_\ggo(\bt)\,\Pfp\Big(\Tb(\{y^a\},\{x^s\})\,\BF(\bt)\Big)\rvac\Big|_{y^a=x^s=z}
\end{split}
\end{equation}
to the reordering of the currents in the product 
\begin{equation}\label{reod}
\begin{split}
\Tb(\{y^a\},\{x^s\})\,\BF(\bt)&=
F_{n,n'}(\{y^a\},\{x^s\})\,k^+_n(y^{n-1})\,\BF(\bt)=\\
&=f(y^{n-1},\bt^{n-1})\,\mathcal{M}_\ggo(\{y^a\},\{x^s\},\bt)\,
\BF(\bbw)\,k_n^+(y^{n-1})\,.
\end{split}
\end{equation}
The factor $f(y^{n-1},\bt^{n-1})$ in \r{reod} appears 
due to the commutation relation 
\begin{equation}\label{kBF}
k^+_n(y)\,\BF(\bt)=
f(y,\bt^{n-1})\,\BF(\bt) \,k^+_n(y)
\end{equation}
which is obtained from the $RTT$ commutation relations \r{rtt}.

In \r{reod} the sets $\bbw^s$, $s\in\PSR_\ggo$ are 
\begin{equation}\label{pre-w}
\bbw^s=\begin{cases}
\ \,\{\bt^s,y^s\},&\ggo=\gga\,,\\
\begin{array}{ll}
\{\bt^s,y^s\},&0\leq s<\varphi_\ggo,\\
\{\bt^s,y^s,x^s-c(s+\theta_\ggo)\},&\varphi_\ggo\leq s\leq n-1,
\end{array}
&\ggo\not=\gga\,,
\end{cases}
\end{equation}
and the normalization factor $\mathcal{M}_\ggo(\{y^a\},\{x^s\},\bt)$ is equal to

\begin{equation}\label{norM-gln}
\mathcal{M}_\ggo(\{y^a\},{\{x^s\}},\bt)=
\prod_{s=1}^{n-2}f(\{y^{s+1},\bt^{s+1}\},y^s)^{-1}
\quad\mbox{for}\quad \ggo=\gga,
\end{equation} 
\begin{equation}\label{norM-g}
\begin{split}
\mathcal{M}_\ggo(\{y^a\},\{x^s\},\bt)&=
\prod_{s=0}^{\varphi_\ggo-1}
f_s(\{y^{\varphi_\ggo},\bt^{\varphi_\ggo}\},y^s)^{-1}
\times\\
&\times \prod_{s=\varphi_\ggo}^{n-2}
f(\{y^{s+1},\bt^{s+1}\},\{y^s,x^s-c(s+\theta_\ggo)\})^{-1},
\quad\mbox{for}\quad \ggo\not=\gga\,.
\end{split}
\end{equation}

The lemma~\ref{BVdef3} follows from the equalities \r{red-ac} 
and the fact that for $a,s\in\PSR_\ggo$ and $\varphi_\ggo\leq s\leq n-1$
\begin{equation}\label{last-ver}
\mathcal{M}_\ggo(\{y^a\},\{x^s\},\bt)\, \mathcal{N}_\ggo(\bbw)^{-1}
\Big|_{y^a=x^s=z}=f(z,\bt^{n-1})^{-1}\,\mathcal{N}_\ggo(\bt)^{-1}\,\mu^n_{n'}(z;\bt)
\end{equation}
is non-singular and yields the rational function $\mu^n_{n'}(z;\bt)$ given 
by \r{mus-all}. The collection of the sets $\bbw$ in the l.h.s. of the equality
\r{last-ver} is given by \r{pre-w} and becomes the collection of the 
extended sets $\bw$ \r{exBCD} in the limit $y^a=x^s\to z$.\qed

\subsection{Zero modes action}

For any two simple labels  $a$ and $b$ of the Dynkin diagram for the algebra
 $\ggo$
  we define a rational function of two variables 
\begin{equation}\label{gam-all}
\bgam^b_a(u,v)=
\begin{cases}
g(v+c\,\nusp_a,u)^{-1},
&b>a,\quad \fb_{a,b}\not=0\,,\\[2mm]
g(u,v)^{-1},&b<a,\quad \fb_{a,b}\not=0\,,\\[2mm]
\displaystyle \bgam_a(u,v),&
a=b,\quad \fb_{a,b}\not=0\,,\\[2mm]
1,&a\not=b,\quad\fb_{a,b}=0.
\end{cases}
\end{equation}

For $a\in\PSR_\ggo$ we define the functions 
$\alpha_a(z)$ 
\begin{equation}\label{alpha-all}
\hC^+_a(z)\,\rvac= \nusp_a\,\alpha_a(z)\,\rvac
\end{equation}
and the generators of the nilpotent 
part of the algebra $\ggo$  expressed in terms of the zero modes 
of the monodromy matrix as follows 
\begin{equation}\label{zm-sr}
\Tzm_a=\TT_{a+\nu_a,a}=\EE^+_{a,a+\nu_a}[0]\,.
\end{equation}
\begin{prop}\label{BVdef1}
The action of the zero mode operator $\Tzm_a$, $a\in \PSR_\ggo$,
on the  vector $\Bf(\bt)$ is given by 
\begin{equation}\label{bvde1}
\Tzm_a\,\Bf(\bt)=\sum_{\rm part}
\Big(\alpha_a(\bt^a_{\so})\,\Oml_a(\bt_{\st},\bt_{\so})-
\Omr_a(\bt_{\so},\bt_{\st})\Big)\Bf(\bt_{\st})\,,
\end{equation}
where the sum runs over partitions $\bt^s\dashv\{\bt^s_{\so},\bt^s_{\st}\}$ with 
cardinality $|\bt^s_{\so}|=\delta_{s,a}$ and
\begin{equation}\label{bvde2}
\Omr_a(\bt_{\so},\bt_{\st})=\prod_{b\,\in\,\PSR_{\ggo}}
\bgam_a^b(\bt^a_{\so},\bt^b_{\st}),\qquad
\Oml_a(\bt_{\st},\bt_{\so})=\prod_{b\,\in\,\PSR_{\ggo}}
\bgam_b^a(\bt^b_{\st},\bt^a_{\so})\,.
\end{equation}
\end{prop}
\proof According 
to the expansion \r{T-eL}, the Gaussian decomposition \r{Gauss} and the 
Ding-Frenkel formulas \r{sr-cur}  the zero mode $\Tzm_a=\EE^+_{a,a+\nus_a}[0]$
coincides with the zero mode $E_a[0]$ of the simple root current $E_a(u)$
for $a\in\PSR_\fg$.

Let $\IE$ be the right ideal of $Y(\ggo)$ 
\begin{equation}\label{Jid}
\IE=Y(\ggo)\, \Big(U^+_E\cap {\rm Ker}(\varepsilon)\Big)\,.
\end{equation}
Here $U^+_E$ is the subalgebra of the Yangian $Y(\ggo)$ 
generated by  the Gaussian coordinates  $\EE^+_{i,j}(u)$
and $\varepsilon$ is the counit. It is clear  that the  ideal $\IE$ annihilates
the vacuum vector $\rvac$ due to \r{rvec1}. 

The commutation relation \r{EF} yields 
\begin{equation}\label{EFcomC}
[E_{a}[0],F_{b}(v)]= 
\nusp_a\ \delta_{a,b}\ \Big(k^{+}_{a}(v)\cdot 
k^{+}_{a+\nus_a}(v)^{-1}-k^{-}_{a}(v)\cdot k^{-}_{a+\nus_a}(v)^{-1}\Big)
\end{equation}
and 
\begin{equation}\label{EFczmC}
 [E_{a}[0],\FF_{b+\nus_b,b}^{-}(v)]= 
 \nusp_a\ \delta_{a,b}\ \Big(k^{-}_{a}(v)\cdot 
 k^{-}_{a+\nus_a}(v)^{-1}-1\Big)\,.
\end{equation}

In order to calculate the action of the zero mode operators $\Tzm_a$ 
 on the off-shell Bethe vectors $\Bf(\bar t)$
 one can use the normal ordering relation
\r{nor-form}  to present the  normalized projection of  
$\mathcal{N}_\fg(\bt)\,\BF(\bt)$
 in the form \r{pap1}:
\begin{equation}\label{prC}
 \mathcal{N}_\fg(\bt)\,\Pfp(\BF(\bar{t})) = \mathcal{N}_\fg(\bt)\,\BF(\bar{t}) + \sum_{a\in\PSR_\ggo}
 \sum_{\rm part}\nusp_a^{-1}\ \Omr_a(\bt_\so,\bt_\st)\,
  \FF^{-}_{a+\nus_a,a}(\bt^{a}_{\so}) \cdot 
 \mathcal{N}_\fg(\bt_\st)\, \BF(\bt_{\st})  + \mathfrak{F}\,,
\end{equation}
where the summation runs over partitions 
  such that the cardinalities of the subsets $\bt^s_{\so}$ 
 are equal to $\delta_{s,a}$.
 The collection of sets $\bt_\st$ in \r{prC}  is equal to 
 $\{\bt^1,\ldots,\bt^{a-1},\bt^a_\st,\bt^{a+1},\ldots,\bt^{n-1}\}$ and  
 the element $\mathfrak{F}\in\IF$ satisfies the property 
 \begin{equation}\label{Fprop}
 [E_a[0],\mathfrak{F}]\in\IF\,.
 \end{equation}

Commuting left and right hand sides of the equality \r{prC} with a zero mode 
operator $E_a[0]$ one gets 
\begin{equation}\label{prC1}
\begin{split}
& \mathcal{N}_\fg(\bt)\,\Big[E_a[0],\Pfp\Big(\BF(\bar{t})\Big)\Big] = \mathcal{N}_\fg(\bt)\,
 \Big[E_a[0],\BF(\bar{t})\Big] +\\
 &\quad + 
 \sum_{\rm part}\nusp_a^{-1}\ \Omr_a(\bt_\so,\bt_\st)
 \Big[E_a[0], \FF^{-}_{a+\nus_a,a}(\bt^{a}_{\so})\Big] \cdot 
 \mathcal{N}_\fg(\bt_\st)\, \Pfp\Big(\BF(\bt_{\st})\Big)  + \mathfrak{F}'\,,
 \end{split}
\end{equation}
where again the element $\mathfrak{F}'\in\IF$.

Note now that the l.h.s. of the equality \r{prC1} belongs to the 
Yangian subalgebra $Y(\fg)\subset \RDYg$. 
It means that all the terms in the r.h.s. of \r{prC1} 
which belong to the left ideal $\IF$ or contain the ratio 
of the 'negative' Cartan currents $k^{-}_{a}(\bt^a_\so)\, 
 k^{-}_{a+\nus_a}(\bt^a_\so)^{-1}$ should vanish. 
 Neglecting these terms in the r.h.s. of \r{prC1} and using 
 \r{EFcomC}, \r{EFczmC} and the commutation relations \r{KF} 
 one gets from \r{prC1} the following equality in the Yangian  $Y(\fg)$
 \begin{equation}\label{prC2}
\begin{split}
& \mathcal{N}_\fg(\bt)\,\Big[E_a[0],\Pfp\Big(\BF(\bar{t})\Big)\Big] =  \\
&\qquad =\sum_{\rm part} 
 \mathcal{N}_\fg(\bt_\st)\, \Pfp\Big(\BF(\bt_{\st})\Big)\,
 \Big(\Oml_a(\bt_\st,\bt_\so)\, 
 k^+_a(\bt^a_\so)\,k^{+}_{a+\nus_a}(\bt^a_\so)^{-1} - \Omr_a(\bt_\so,\bt_\st) \Big)\,.
 \end{split}
\end{equation}
Applying this equality to the vacuum vector $\rvac$ one gets \r{bvde1}.
 \qed

 This ends the proof of the main theorem~\ref{main-thm}. 
From now on, we will identify the vectors 
$\Bf(\bt)$ given by  \r{BVg}  
with the off-shell Bethe vectors $\BB(\bt)$ corresponding to the definition \ref{def:offBV}.

\subsection{Normalization of off-shell Bethe vectors}

The defining relations for the off-shell Bethe vectors not only fix 
the structure of polynomial $\mathcal{P}$ 
which defines the off-shell Bethe vector \r{bvgr}
 but also fix, when the functions 
$\mus^n_{n'}(z;\bt)$ \r{mus-all}, $\Oml(\bt_{\st},\bt_{\so})$, 
$\Omr(\bt_{\so},\bt_{\st})$ \r{bvde2} 
are given,  the overall normalization 
of these vectors. In the framework of the projection method this 
unique normalization is given by the prefactor in \r{BVg}
\begin{equation}\label{dif-form}
\prod_{s\in\PSR_{\ggo}}\frac{1}{g(\bt^{s+\nus_s},\bt^s)}
\prod_{\ell_1>\ell_2}^{r_s}\bgam_s(t^s_{\ell_1},t^s_{\ell_2})
=
\prod_{s\in\PSR_{\ggo}}\frac{\nusp_s^{|\bt^s|}}
{g(\bt^{s+\nus_s},\bt^s)\,h_s(\bt^s,\bt^s)^{\zeta_s}}
\prod_{\ell_1>\ell_2}^{r_s}f_s(t^s_{\ell_1},t^s_{\ell_2})\,.
\end{equation}

As it was already mentioned at the end of the section~\ref{g-all}
we will always consider the products of  simple root currents 
in the category of the highest  weight
representations of the Yangian double $\DYg$.
In this category, arbitrary matrix element of such a product of currents 
between generic vectors 
from the highest  weight representation may be considered as the function 
with certain analytical properties dictated by the Serre and commutation relations 
between  simple root currents. In this context the formal parameters 
of the currents maybe replaced by  complex numbers and identified 
with the Bethe parameters $t^s_\ell$, $\ell=1,\ldots,|\bt^s|$, $s\in\PSR_\ggo$
(see \cite{EKhP} for such an interpretation in the case of the quantum 
affine algebras). 

We will demonstrate below that in this approach, the 
projection of the product of  simple root currents in \r{BVg}, considered as 
a function the Bethe parameters, may have zeros and poles. 
A unique feature of the normalization factor \r{dif-form} is that 
most of the  zeros and poles 
are compensated by this factor. The price to pay 
 is the apparition of spurious poles coming from this normalization 
factor. Let us illustrate this phenomena in the simplest 
example of the off-shell Bethe vector in $\mathfrak{sl}_2$-invariant 
integrable model. The off-shell Bethe vector $\BB(\bt)$ in this case 
depends on a single set of  Bethe parameters $\bt=t_1,\ldots,t_{|\bt|}$ 
and the defining relation  \r{T-act-g} for the off-shell Bethe vectors
 reads 
\begin{equation}\label{nor1}
T_{1,2}(z)\,\BB(\bt)=\lambda_2(z)\,h(\bt,z)\,h(z,\bt)\,\BB(\{\bt,z\})\,.
\end{equation}
This relation can be resolved (up to the overall factor not depending on the 
Bethe parameters) in the well-known form 
\begin{equation}\label{nor2}
\BB(\bt)=\frac{1}{\lambda_2(\bt)\,h(\bt,\bt)}\
T_{1,2}(t_1)\,T_{1,2}(t_2)\cdots T_{1,2}(t_{|\bt|})\rvac\,.
\end{equation}
Since according to the $RTT$ commutation relations, the 
monodromy entries $T_{1,2}(u)$ commute for 
different values of the spectral parameters, there is nothing special 
in the product $T_{1,2}(\bt)$ of monodromy entries in \r{nor2} when 
$t_i=t_j\pm c$, $i\not=j$. On the other hand the denominator 
of the overall normalization factor in \r{nor2}
will generate a simple poles for these relations between Bethe parameters.
We will call such type of singularities, {\sl spurious singularities}.
They can be always removed by a renormalization of the particular 
off-shell Bethe vectors. 
In the appendix~\ref{AppB} we will provide examples of spurious 
singularities for $\mathfrak{sl}_3$ type off-shell Bethe vectors and 
 for   $\mathfrak{sp}_4$ type Bethe vectors.

\subsection{Coproduct property}

It was proved in the paper \cite{LPR-RR} that  
the off-shell Bethe vectors of the definition \r{def:offBV}, satisfy a coproduct property. 
 This statement can be verified independently 
in the framework of the projection method.

Let us consider the composed integrable 
model \cite{IK84, S22} whose monodromy matrix 
 is defined by the standard 
coproduct \r{st-cop} in the Yangian double subalgebra $Y(\ggo)$,
$\Delta(T(z))=T^{[2]}(z)\cdot T^{[1]}(z)$, where $(A\cdot B)$
stands for
 the matrix multiplication 
in $N_\ggo$-dimensional auxiliary space.
For any partition $\bt^s\dashv\{\bt^s_{\so},\bt^s_{\st}\}$, $s\in\PSR_\ggo$,  we define 
the function  
\begin{equation}\label{Om-gen}
\Om(\bt_{\so},\bt_{\st})=\prod_{a\,\in\,\PSR_{\ggo}}\Omr_a(\bt_{\so},\bt_{\st})
=\prod_{a\,\in\,\PSR_{\ggo}}\Oml_a(\bt_{\so},\bt_{\st})=
\prod_{a,b\,\in\,\PSR_{\ggo}}\bgam^b_a(\bt^a_{\so},\bt^b_{\st})\,.
\end{equation}
Then the off-shell Bethe vectors 
$\BB(\bt)$ \r{BVg} satisfy the following coproduct property.
\begin{prop}\label{cop-pro}
The off-shell Bethe vectors of the composed model can be expressed 
as  
\begin{equation}\label{BV-g-cop}
\Delta\,\BB(\bt)=\sum_{{\rm part}\ \bt} 
\Om(\bt_{\so},\bt_{\st})\,
\prod_{s\in\PSR_\ggo}\alpha^{[2]}_s(\bar t^s_{\st})
\, \BB^{[1]}(\bar t_{\st})\, \BB^{[2]}(\bar t_{\so})\,. 
\end{equation}
\end{prop}
\proof 
Let  $\rvac=\rvac^{[1]}\otimes\rvac^{[2]}$ be the vacuum vector 
of the composed $\ggo$-invariant integrable model. 
Using the relation between the coproduct $\Delta$ \r{st-cop} in Yangian double 
 $\RDYg$ and the Drinfeld's coproduct \r{D-cop} in the 
completed subalgebra  $\oUF$,
and exploring 
the same approach as in the proof of the theorem~4 
in \cite{EKhP}, we can show that for any element  $\cF\in \oUF$ 
\begin{equation}\label{mag}
\Delta\Big(\Pfp\big(\cF\big)\Big)=(\Pfp\otimes\Pfp)\,
\Delta^{(D)}(\cF)\quad {\rm mod}\quad  U^+_F\otimes \IE\,,
\end{equation}
where $\IE$ is the right ideal of $Y(\ggo)$ defined by \r{Jid}.
Then the proof of 
the coproduct property is a straightforward calculation using the 
coproduct \r{D-cop} in  $\oUF$ 
and the commutation relations \r{KF} and \r{FF}.\qed

\section{Analytical properties of off-shell Bethe vectors}
\label{an-pro-BV}

For any two different nodes $a\neq b$ of the Dynkin diagram of $\ggo$ 
such that $\fb_{a,b}\neq 0$ we consider the cardinality 1 set $\bt^b=\{t^b\}$
and the cardinality $m_{a,b}=1-\ago_{a,b}$ set $\bt^a=\{t^a_1,\ldots,t^a_{m_{a,b}}\}$,
where $\ago_{a,b}$ is the Cartan matrix for the algebra $\ggo$. 
Let us consider the following  element from the subalgebra  $\oUF$
\begin{equation}\label{Ser-ele}
\begin{split}
\cF_{a,b}(\bt^a,\bt^b)&=
\prod_{\ell_1<\ell_2}^{m_{a,b}}\cre_{a,a}(t^a_{\ell_1},t^a_{\ell_2})\ 
\cre_{b,a}(t^b,\bt^a)\ F_b(t^b)\,
F_a(t^a_1)\cdots F_a(t^a_{m_{a,b}})=\\
 &=
\prod_{\ell_1<\ell_2}^{m_{a,b}}\cre_{a,a}(t^a_{\ell_1},t^a_{\ell_2})\ 
\cre_{a,b}(\bt^a,t^b)\ 
F_a(t^a_1)\cdots F_a(t^a_{m_{a,b}})\, F_b(t^b)\,.
 \end{split}
\end{equation}
The equality between the first and the second line in \r{Ser-ele} as well as the
symmetry  of this element  with respect to permutations 
in the set $\bt^a$ follow from the commutation relations 
\r{FF}.

Using results of the paper \cite{LPR-SR-CC}
 we can describe the analytical properties 
of off-shell Bethe vectors \r{BVg}. We will describe these analytical properties 
separately for each algebra $\ggo$.  

\subsection{Analytical properties of $\gga$ type Bethe vectors}\label{an-pr-gga}

Let us fix a positive integer $s\in\PSR_{\gga}$.  
Using the definition  \r{BVg} 
 it is clear that Bethe vectors contain elements $\cF_{s,s+1}(\{t^s_{i_1},t^s_{i_2}\},t^{s+1}_j)$ 
 for  $i_1\not=i_2$ such that $i_1,i_2=1,\ldots,|\bt^s|$, $j=1,\ldots,|\bt^{s+1}|$ 
and $s=1,\ldots,n-2$. 
Analogously, for $s=2,\ldots,n-1$ there are elements
$\cF_{s,s-1}(\{t^{s}_{j_1},t^{s}_{j_2}\},t^{s-1}_{i})$ for 
$j_1\not=j_2$ such that $i=1,\ldots,|\bt^{s-1}|$, $j_1,j_2=1,\ldots,|\bt^{s}|$. 
According to the  proposition~3.1 in \cite{LPR-SR-CC},
the element $\cF_{s,s+1}(\{t^s_{i_1},t^s_{i_2}\},t^{s+1}_j)$
 vanishes on the hyperplanes\footnote{These type hyperplanes were 
 called the Serre stratums in the paper \cite{DKhP}.}
\begin{equation}\label{h-mi1}
t^{s+1}_j=t^s_{i_1}=t^s_{i_2}-c\,,\quad t^{s+1}_j=t^s_{i_2}=t^s_{i_1}-c\,,\quad
s=1,\ldots,n-2\,,
\end{equation}
for $i_1<i_2$ and is regular on the hyperplanes 
\begin{equation}\label{h1reg}
t^{s+1}_j=t^s_{i_1}=t^s_{i_2}+c\,,\quad t^{s+1}_j=t^s_{i_2}=t^s_{i_1}+c\,,\quad
s=1,\ldots,n-2\,.
\end{equation}

Analogously, according to  the proposition~3.2 in \cite{LPR-SR-CC},
the element $\cF_{s,s-1}(\{t^{s}_{j_1},t^{s}_{j_2}\},t^{s-1}_{i})$
vanishes on the hyperplanes
\begin{equation}\label{h-mi2}
t^{s-1}_i=t^{s}_{j_1}=t^{s}_{j_2}+c\,,\quad t^{s-1}_i=t^{s}_{j_2}=t^{s}_{j_1}+c\,,
\quad s=2,\ldots,n-1
\end{equation}
for $j_1<j_2$ and is regular on the hyperplanes 
\begin{equation}\label{h2reg}
t^{s-1}_i=t^{s}_{j_1}=t^{s}_{j_2}-c\,,\quad t^{s-1}_i=t^{s}_{j_2}=t^{s}_{j_1}-c\,,
\quad s=2,\ldots,n-1\,.
\end{equation}

These zeros which originates from the hyperplanes 
\r{h-mi1} and \r{h-mi2}  are 
canceled by the poles of the normalization factor in these vectors. 
Indeed, this normalization factor  contains the products 
\begin{equation}\label{n-f-mi2}
h(t^{s}_{i_2},t^s_{i_1})^{-1}\,h(t^{s}_{i_1},t^s_{i_2})^{-1}\quad
\mbox{and}\quad 
h(t^{s}_{j_1},t^{s}_{j_2})^{-1}\,h(t^{s}_{j_2},t^{s}_{j_1})^{-1}
\end{equation}
both appearing from the product $h(\bt^{s},\bt^s)$. 
The factors in \r{n-f-mi2} have simple poles on the 
hyperplanes \r{h-mi1} and \r{h-mi2} due to the equalities 
$t^s_{i_1}=t^s_{i_2}-c$ or $t^s_{i_2}=t^s_{i_1}-c$ and 
 $t^s_{j_1}=t^s_{j_2}+c$ or $t^s_{j_2}=t^s_{j_1}+c$
respectively. 

The normalization of the off-shell Bethe vectors \r{BVg} 
induces the apparition of {\sl spurious} singularities  
 of the off-shell Bethe vectors \r{BVg}. For example, 
the Bethe vectors 
\begin{equation}\label{sin-bv}
\BB(\{\bt^a,z\}_1^{n-2},\{\bt^{n-1},z+c,z\})\quad\mbox{and}\quad 
 \BB(\{\bt^1,z-c,z\},\{\bt^a,z\}_2^{n-1})
\end{equation}
become singular due to the factors $h(\bt^{n-1},\bt^{n-1})^{-1}$ and 
$h(\bt^{1},\bt^{1})^{-1}$ in the normalization of \r{BVg}. 
The Bethe vectors of the type \eqref{sin-bv} appear in
 the action of the monodromy entry $T_{1,n}(z)$ 
on the non-singular Bethe vectors $\BB(\{\bt^a\}_1^{n-2},\{\bt^{n-1},z+c\})$ and 
 $\BB(\{\bt^1,z-c\},\{\bt^a\}_2^{n-1})$:
\begin{equation}\label{T-act-n}
T_{1,n}(z)\,\BB(\{\bt^a\}_1^{n-1})=\lambda_n(z)\, h(\bt^1,z)\, h(z,\bt^{n-1})\,
\BB(\{\bt^a,z\}_1^{n-1})\,.
\end{equation}
However, the pre-factor in  \r{T-act-n}   compensates these spurious singularities.
Overall, the action is well-defined, and one can get a renormalized version of the 
Bethe vectors \r{sin-bv}  which does not contain any singularities.

In the appendix~\ref{AppB} we present examples of $\mathfrak{sl}_3$ type 
off-shell Bethe vectors on the Serre hyperplanes expressed in terms 
of monodromy matrix entries $T_{i,j}(u)$ and illustrate the mechanism 
of removing these spurious singularities. In what follows we 
will not pay attention to these spurious  singularities of the Bethe vectors 
\r{BVg}, keeping in mind that they are removable.

\subsection{Analytical properties of $\ggd$ type  Bethe vectors}
\label{an-pr-ggd}

The mechanism of compensation of the spurious singularities for 
$\mathfrak{o}_{2n}$ type Bethe vectors will be slightly different 
because the action \r{T-act-n} is replaced by the action 
\begin{equation}\label{T-act-D}
T_{n',n}(z)\,\BB(\{\bt^s\}_0^{n-1})=\lambda_n(z)\, \mu^n_{n'}(z;\bt)\ 
\BB(\{\bt^s,z\}_0^1,\{\bt^s,z,z_s\}_2^{n-1})\,,
\end{equation}
where  
\begin{equation}\label{mu-hD2}
\mu^n_{-n+1}(z;\bt)=-\ (-1)^{|\bt^0|+|\bt^1|}\ 
(n-1)\ \frac{h(z,\bt^{n-1})}{g(z_n,\bt^{n-1})}\,.
\end{equation}

The ratio $\frac{h(z,\bt^{n-1})}{g(z_n,\bt^{n-1})}$
is present in the function $\mu^n_{n'}(z;\bt)$ to compensate 
the spurious singularity of the Bethe vectors 
\begin{equation}\label{BVDdi1}
\BB(\{\bt^s,z\}_0^{1},\{\bt^s,z,z_s\}_2^{n-2},\{\bt^{n-1},z+c,z,z_{n-1}\})
\end{equation}
and 
\begin{equation}\label{BVDdi2} 
\BB(\{\bt^s,z\}_0^{1},\{\bt^s,z,z_s\}_2^{n-2},\{\bt^{n-1},z,z_{n-1}-c,z_{n-1}\})
\end{equation}
which appears in the r.h.s. of \r{T-act-D} 
after the action of $T_{n',n}(z)$ on the Bethe vectors 
\begin{equation}\label{BVDdi3}
\BB(\{\bt^s\}_0^{n-2},\{\bt^{n-1},z+c\})\quad\mbox{and}\quad
\BB(\{\bt^s\}_0^{n-2},\{\bt^{n-1},z_{n-1}-c\})\,.
\end{equation}
The singular factor  $h(z+\epsilon,z+c)^{-1}\Big|_{\epsilon\to 0}$ 
in the vector \r{BVDdi1} is compensated 
by the zero factor  $h(z,z+c)=0$ coming from $h(z,\{\bt^{n-1},z+c\})$ in the 
normalization function \r{mu-hD2}. Analogously, 
the singular factor  $h(z_{n-1}-c+\epsilon,z_{n-1})^{-1}\Big|_{\epsilon\to 0}$
in the vector \r{BVDdi2} is compensated 
by the zero factor  $g(z_{n},z_{n})^{-1}=0$ 
coming from $g(z_n,\{\bt^{n-1},z_n\})$ in the 
same normalization function \r{mu-hD2}.

The remaining simple pole singularities coming from the normalization factor 
$\prod_{s=0}^{n-1}h(\bt^s,\bt^s)^{-1}$ in the definition of the 
Bethe vectors \r{BVg} are always canceled by the zeros 
originating from the Serre relations \r{SR-g}. The mechanism of  
cancelation is the same as the one explained for 
$\gga$-type  Bethe vectors. 
We illustrate it for the Bethe vector 
\begin{equation}\label{BVDre1}
\BB(\{\bt^0,z\},\{\bt^1,z,z\pm c\},\{\bt^a,z,z_a\}_2^{n-1})\,.
\end{equation}
Let us show that this vector remains nonsingular  when considering the action 
\begin{equation}\label{BVDre3}
T_{n',n}(z)\,\BB(\bt^0,\{\bt^1,z\pm c\},\{\bt^a\}_2^{n-1})\,.
\end{equation}

The Bethe vector 
\begin{equation}\label{BVDre5}
\BB(\{\bt^0,z\},\{\bt^1,z,z+ c\},\{\bt^2,z,z-c\},\{\bt^a,z,z_a\}_3^{n-1})
\end{equation}
is non-singular because the parameter $t^2_j=z$ from the set 
$\{\bt^2,z,z-c\}$ and the parameters $t^1_{i_1}=z$, $t^1_{i_2}=z+c$ 
from the set $\{\bt^1,z,z+ c\}$
satisfy the condition of the Serre hyperplane \r{h-mi1}. Thus, the pole in the normalization 
factor $h(t^1_{i_1},t^1_{i_2})^{-1}$ is compensated by the zero 
of the element $\cF_{1,2}(\{t^1_{i_1},t^1_{i_2}\},t^2_j)$.

Analogously, the Bethe vector (recall that $z_2=z-c$)
\begin{equation}\label{BVDre6}
\BB(\{\bt^0,z\},\{\bt^1,z,z- c\},\{\bt^2,z,z-c\},\{\bt^a,z,z_a\}_3^{n-1})
\end{equation}
is non-singular because the parameter $t^2_j=z-c$ from the set 
$\{\bt^2,z,z-c\}$ and the parameters $t^1_{i_1}=z$, $t^1_{i_2}=z-c$ 
from the set $\{\bt^1,z,z- c\}$ again 
satisfy condition of the Serre hyperplane \r{h-mi1}. 
The singularity in the normalization 
factor $h(t^1_{i_2},t^1_{i_1})^{-1}$ is again compensated by the zero 
of the element $\cF_{1,2}(\{z,z-c\},z-c)$ \r{Ser-ele}.
It is clear from this consideration that the cancelation of the pole and zero 
due to the Serre relation appears in this case because the set 
$\bt^2$ in the Bethe vector \r{BVDre3} is extended also by $z_2=z-c$ that 
never happen in the case of $\gga$-type Bethe vectors. 

The fact that the Bethe vector \r{BVDre1} is regular 
on the hyperplanes 
\begin{equation}\label{sih1}
t^2_j=t^1_{i_1}=t^1_{i_2}-c\quad\mbox{and}\quad 
t^2_j=t^0_{i_1}=t^0_{i_2}-c
\end{equation}
is confirmed also by the normalization factor  \r{mu-hD2}
since it does not depend on the sets $\bt^0$ and $\bt^1$ but only on their 
cardinalities.

\subsection{Analytical properties of $\gsp$ off-shell Bethe vectors}
\label{an-pr-ggc}

In the case of $\gsp$-invariant integrable models, the action of the highest monodromy 
reads
\begin{equation}\label{T-act-C}
T_{n',n}(z)\,\BB(\{\bt^s\}_0^{n-1})=\lambda_n(z)\,\mu^n_{n'}(z;\bt)\ 
\BB(\{\bt^0,z\},\{\bt^s,z,z_s\}_1^{n-1})\,,
\end{equation}
where $z_s=z-c(s+1)$ and 
the normalization function 
\begin{equation}\label{mu-hC2}
\mu^n_{n'}(z;\bt)=-\  (-1)^{|\bt^0|}\ (n+1)\ h(\bt^1,z)\ 
\frac{h(z,\bt^{n-1})}{g(z_n,\bt^{n-1})}\,.
\end{equation}
Note that this normalization factor contains both the 
 ratio $\frac{h(z,\bt^{n-1})}{g(z_n,\bt^{n-1})}$ occurring for the $\mathcal{D}Y(\goN)$ Bethe vectors 
 and the factor $h(\bt^1,z)\,h(z,\bt^{n-1})$ which appeared in the $\DYSn$ case. 

As for $\ggd$ type Bethe vectors 
the ratio $\frac{h(z,\bt^{n-1})}{g(z_n,\bt^{n-1})}$ in \r{mu-hC2} 
  compensates
the spurious singularities of the Bethe vector 
\begin{equation}\label{BVCdi1}
\BB(\{\bt^0,z\},\{\bt^s,z,z_s\}_1^{n-2},\{\bt^{n-1},z+c,z,z_{n-1}\})
\end{equation}
and 
\begin{equation}\label{BVCdi2} 
\BB(\{\bt^0,z\},\{\bt^s,z,z_s\}_1^{n-2},\{\bt^{n-1},z,z_{n-1}-c,z_{n-1}\})
\end{equation}
which appear in the r.h.s. of \r{T-act-C} 
after the action of $T_{n',n}(z)$ on the Bethe vectors 
\begin{equation}\label{BVCdi3}
\BB(\{\bt^s\}_0^{n-2},\{\bt^{n-1},z+c\})\quad\mbox{and}\quad
\BB(\{\bt^s\}_0^{n-2},\{\bt^{n-1},z_{n}\})
\end{equation}
due to the normalization 
factor $h(\bt^{n-1},\bt^{n-1})^{-1}$ in \r{BVg}.  
The singular factor  $h(z+\epsilon,z+c)^{-1}\Big|_{\epsilon\to0}$ 
in the vector \r{BVCdi1} is compensated 
by the zero factor  $h(z,z+c)=0$ coming from 
$h(z,\{\bt^{n-1},z+c\})$ in the 
normalization function \r{mu-hC2}. Analogously, 
the singular factor  $h(z_{n-1}-c+\epsilon,z_{n-1})^{-1}\Big|_{\epsilon\to0}$ 
in the vector \r{BVCdi2} is compensated 
by the zero factor  $g(z_n,z_{n})^{-1}=0$ 
coming from $g(z_n,\{\bt^{n-1},z_n\})^{-1}$ in the 
same normalization function \r{mu-hC2}.
The remaining simple pole singularities coming from the normalization factor 
$h_0(\bt^0,\bt^0)\prod_{s=1}^{n-1}h(\bt^s,\bt^s)^{-1}$ 
 are always canceled by the zeros 
originating from the Serre relations. 
The mechanism of this 
cancelation is the same as the one detailed for 
$\gga$ type 
Bethe vectors.

Let us illustrate it for the  Bethe vectors 
\begin{equation}\label{BVCre2}
\BB(\{\bt^0,z,z\pm 2\,c\},\{\bt^a,z,z_a\}_1^{n-1})
\end{equation} 
which can be obtained by the action 
\begin{equation}\label{BVCre4}
T_{n',n}(z)\,\BB(\{\bt^0,z\pm 2\,c\},\{\bt^a\}_1^{n-1})\,.
\end{equation}

The Bethe vector (recall that $z_1=z-2\, c$)
\begin{equation}\label{BVCre5}
\BB(\{\bt^0,z,z+ 2\,c\},\{\bt^1,z,z-2\,c\},\{\bt^a,z,z_a\}_2^{n-1})
\end{equation}
is non-singular because the parameter $t^1_j=z$ from the set 
$\{\bt^1,z,z-2\,c\}$ and the parameters $t^0_{i_1}=z$ and  
$t^0_{i_2}=z+2\,c$ 
from the set $\{\bt^0,z,z- 2\,c\}$
are on the Serre hyperplane $t^1_j=t^0_{i_1}=t^0_{i_2}-2\,c$. 
The pole in the normalization 
factor $\hc(t^0_{i_1},t^0_{i_2})^{-1}$ is compensated by the zero 
of the element $\cF_{0,1}(\{z,z+2\,c\},z)$.

Analogously, the Bethe vector (recall again that $z_2=z-2\,c$)
\begin{equation}\label{BVCre6}
\BB(\{\bt^0,z,z-2\,c\},\{\bt^1,z,z- 2\,c\},\{\bt^a,z,z_a\}_2^{n-1})
\end{equation}
is non-singular because the parameter $t^1_j=z-2\,c$ from the set 
$\{\bt^1,z,z-2\,c\}$ and the parameters $t^0_{i_1}=z-2\,c$, 
$t^0_{i_2}=z$ 
from the set $\{\bt^0,z,z-2\,c\}$ again 
are on the Serre hyperplane $t^1_j=\tilde t^0_{i_1}=\tilde t^0_{i_2}-2\,c$. 
It is clear from this consideration that the cancellation of the poles and zeroes 
due to the Serre relation  occurs because the set 
$\bt^1$ in the Bethe vector \r{BVCre6} is extended by $z_1=z-2\,c$. 

Using similar arguments we can show that the Bethe vector
\begin{equation}\label{BVCsi1}
\BB(\{\bt^0,z\},\{\bt^1,z,z-c,z-2\,c\},\{\bt^s,z,z_s\}_2^{n-1})
\end{equation} 
 has a removable spurious singularity. Indeed, according to the definition
\r{BVg} it contains the zero element $\cF_{1,0}(\{z,z-c,z-2\,c\},z)=0$ 
due to the Serre relation   but 
a double pole coming from normalization due to $h(z-c,z)=h(z-2\,c,z-c)=0$.
The Bethe vector \r{BVCsi1} can be obtained 
from the action 
\begin{equation}\label{BVCsi2}
\begin{split}
&T_{n',n}(z)\,\BB(\bt^0,\{\bt^1,z-c\},\{\bt^s\}_2^{n-1})
=-\ \lambda_n(z)\,(-1)^{|\bt^0|}\,h(\{t^1,z-c\},z)\,\kappa_n\,
\frac{h(z,\bt^{n-1})}{g(z_n,\bt^{n-1})}\times\\
&\qquad\qquad\qquad\qquad\qquad\qquad\qquad\qquad\times 
\BB(\{\bt^0,z\},\{\bt^1,z,z-c,z-2\,c\},\{\bt^s,z,z_s\}_2^{n-1}).
\end{split}
\end{equation}
The r.h.s. of equality \r{BVCsi2} 
is regular because the simple pole of the Bethe vector is compensated by the simple zero coming from the 
normalization factor $h(\{t^1,z-c\},z)$. 

\subsection{Analytical properties of the Bethe vectors of 
$\ggb$-invariant models}\label{an-pr-ggb}

The action of the highest monodromy now takes the form
\begin{equation}\label{T-act-B}
T_{n',n}(z)\,\BB(\{\bt^s\}_0^{n-1})=\lambda_n(z)\,\mu^n_{n'}(z;\bt)\ 
\BB(\{\bt^s,z,z_s\}_0^{n-1})\,,
\end{equation}
where $z_s=z-c(s-1/2)$ and 
the normalization function 
\begin{equation}\label{mu-hB2}
\mu^n_{n'}(z;\bt)=-\  \frac{g(z_1,\bt^0)}{h(z,\bt^0)}\, 
\frac{h(z,\bt^{n-1})}{g(z_n,\bt^{n-1})}\,.
\end{equation}
Again, this normalization factor has 
the ratio $\frac{h(z,\bt^{n-1})}{g(z_n,\bt^{n-1})}$ 
occurring in $\ggd$ and $\gsp$ type Bethe vectors, 
which compensate the spurious singularities of 
the off-shell Bethe vectors 
\begin{equation}\label{BVBdi1}
\BB(\{\bt^s,z,z_s\}_0^{n-2},\{\bt^{n-1},z+c,z,z_{n-1}\})
\end{equation}
and 
\begin{equation}\label{BVBdi2} 
\BB(\{\bt^s,z,z_s\}_0^{n-2},\{\bt^{n-1},z,z_{n-1}-c,z_{n-1}\}).
\end{equation}
These vectors appear 
in the action of $T_{n',n}(z)$ on the Bethe vectors 
\begin{equation}\label{BVBdi3}
\BB(\{\bt^s\}_0^{n-2},\{\bt^{n-1},z+c\})\quad\mbox{and}\quad
\BB(\{\bt^s\}_0^{n-2},\{\bt^{n-1},z_{n}\}).
\end{equation}
The singular factor  $h(z+\epsilon,z+c)^{-1}\Big|_{\epsilon\to 0}$
 in the vector \r{BVBdi1} is compensated 
by the zero factor  $h(z,z+c)$ coming from 
$h(z,\{\bt^{n-1},z+c\})$ in the 
normalization function \r{mu-hB2}. Analogously, 
the singular factor  $h(z_{n-1}-c+\epsilon,z_{n-1})^{-1}\Big|_{\epsilon\to 0}$ 
in the vector \r{BVBdi2} is compensated 
by the zero factor  $g(z_n,z_{n})^{-1}$ 
coming from $g(z_n,\{\bt^{n-1},z_n\})^{-1}$ in the 
same normalization function \r{mu-hB2}.
The remaining simple pole singularities coming from the normalization factor 
$\prod_{s=1}^{n-1}h(\bt^s,\bt^s)^{-1}$ 
 are  canceled by the zeros 
originating from the Serre relations except the zero described 
by the proposition~3.3 in \cite{LPR-SR-CC}.  

According to this proposition the $\ggb$ type off-shell Bethe vectors 
\begin{equation}\label{BVBv1}
\BB(\{\bt^0,z-c/2,z,z+c/2\},\{\bt^1,z,z-c/2\},\{\bt^s,z,z_s\}_2^{n-1})
\end{equation}
and 
\begin{equation}\label{BVBv2}
\BB(\{\bt^0,z+c,z,z+c/2\},\{\bt^1,z,z-c/2\},\{\bt^s,z,z_s\}_2^{n-1})
\end{equation}
vanish because parameters $t^0_1=z-c/2$, $t^0_2=z$, 
$t^0_2=z+c/2$ in the sets $\{\bt^0,z-c/2,z,z+c/2\}$ and 
parameter $t^1_1=z-c/2$ in the set 
$\{\bt^1,z,z-c/2\}$ 
satisfy the relations of the Serre hyperplane 
$t^1_j=t^0_{1}=t^0_{2}-c/2=t^0_{3}-c$
while the parameters  in the sets $\{\bt^0,z+c,z,z+c/2\}$ and 
$\{\bt^1,z,z-c/2\}$ also satisfy this condition.  

The vectors \r{BVBv1} and \r{BVBv2}
can be obtained by the action of the highest monodromy entry 
$T_{n',n}(z)$ \r{T-act-B} on the non-singular Bethe vectors 
\begin{equation}\label{BVBnon}
\BB(\{\bt^0,z-c/2\},\{\bt^s\}_1^{n-1})\quad\mbox{and}\quad 
\BB(\{\bt^0,z+c\},\{\bt^s\}_1^{n-1})\,.
\end{equation}
The regularity of  this action follows from the cancelation 
of the zero coming from the Serre relation and the pole of the factor 
 $g(z_1+\epsilon,\{\bt^0,z-c/2\})\Big|_{\epsilon\to0}$ for the first Bethe vector in \r{BVBnon}. 
The same cancelation happens between the pole of the factor 
 $h(z,\{\bt^0,z+c\})^{-1}\Big|_{\epsilon\to0}$ and the zero from the Serre relation for the 
second Bethe vector in \r{BVBnon}. 
This shows that the vanishing of the Bethe vectors \r{BVBv1} and 
\r{BVBv2} is spurious and can be avoided by the renormalization.

\section{Conclusion}

The main result of this paper is to establish that off-shell Bethe vectors in a generic 
$\fg$-invariant integrable model can be defined within the framework of the projection method introduced in~\cite{EKhP} for quantum affine algebras and extended here to Yangian doubles. 
To prove this, it is sufficient to verify that the off-shell Bethe vectors defined by \r{BVg} in the framework of the projection method satisfy the defining relations for off-shell Bethe vectors formulated in \cite{LPR-RR}. In the case of integrable models associated
with the supersymmetric Yangian doubles $\mathcal{D}Y(\mathfrak{gl}(m|n))$,
the projection method was previously explored in~\cite{HLPRS17}.

This result, together with those obtained in the paper \cite{LPR-RR}, shows that the
structure of the space of states in $\fg$-invariant quantum integrable models
can be studied using two equivalent approaches: either with the explicit
construction of off-shell Bethe vectors via the projection method, or with their
implicit definition~\ref{def:offBV}. 

A natural next step is to investigate scalar products of off-shell Bethe
vectors, their norms, and form factors in generic $\fg$-invariant quantum
integrable models using these complementary approaches. 
We leave the exploration of this direction for future work.

\section*{Acknowledgement}
S.P. acknowledges support from the PAUSE Programme and the hospitality of LAPTh, where this work was conducted.
The work of A.~L. was supported by the Beijing Natural Science Foundation (IS24006) and Beijing
Talent Program.

\appendix

\section{Examples of $\fs\fl_3$ Bethe vectors on the Serre hyperplanes}
\label{AppB}

Let us illustrate that the two different approaches to describe the structure of the space 
of states in $\fg$-invariant quantum integrable models yields the same 
formulas for the off-shell Bethe vectors. To do this we consider 
 Bethe vectors in $\fg\fl_3$-invariant model depending on 
 small cardinality sets of Bethe parameters. 

It was shown in \cite{LPR-RR} that 
the definition~\ref{def:offBV} leads to   
the rectangular recurrence relations introduced in \cite{LPR25}. 
These rectangular recurrence relations for the $\fs\fl_3$ type off-shell 
Bethe vector $\BB(u,\{v_1,v_2\})$  
allows to present this vector explicitly as follows
\begin{equation}\label{BVAex}
\begin{split}
\BB(u,\{v_1,v_2\})&=\frac{1}{\lambda_2(u)\,\lambda_3(v_1)\,\lambda_3(v_2)}
\,\frac{1}{g(v_1,u)\,g(v_2,u)}\,\frac{1}{h(v_1,v_2)\,h(v_2,v_1)}\times\\
&\quad\times\Big(T_{2,3}(v_2)\,T_{2,3}(v_1)\,T_{1,2}(u)+\\
&\qquad +\lambda_2(u)\,g(v_1,u)\,f(v_2,v_1)\,T_{1,3}(v_1)\,T_{2,3}(v_2)+\\
&\qquad +\lambda_2(u)\,g(v_2,u)\,f(v_1,v_2)\,T_{1,3}(v_2)\,T_{2,3}(v_1)
\Big)\,|0\rangle\,.
\end{split}
\end{equation}

It was shown in section~\ref{an-pr-gga}  that 
the singularities of normalization factor $h(v_1,v_2)^{-1}\,h(v_2,v_1)^{-1}$ for  
$v_2=v_1-c$ or $v_1=v_2-c$ are cancelled by the Serre relation
on the Serre hyperplanes $u=v_1=v_2+c$ and $u=v_2=v_1+c$. 
There are also spurious singularities at the hyperplanes
$u=v_1=v_2-c$ and $u=v_2=v_1-c$ which have nothing to do with 
the Serre relations and that can be removed through a redefinition of the 
corresponding Bethe vectors.

Let us illustrate these properties from the explicit expression 
\r{BVAex} for the Serre hyperplanes  $u=v_1=v_2+c$ and $u=v_1=v_2-c$. 
When $u=v_1$ 
the first and third terms in \r{BVAex} disappear, due to the factor $\frac1{g(v_1,u)}$, and we get
\begin{equation}\label{forA}
\BB(u,\{u,v_2\})=\frac{T_{1,3}(u)\,T_{2,3}(v_2)|0\rangle}
{\lambda_3(u)\,\lambda_3(v_2)\,h(u,v_2)}\,.
\end{equation}

Assuming that the monodromy matrix entries are generic and are not equal to zero 
or infinity when they act on the vacuum vector we conclude that 
the Bethe vector \r{forA} is non-singular when $v_2\to u- c$
and corresponds to the Serre hyperplane $u=v_1=v_2+c$. 
It leads to 
\begin{equation}\label{A10b}
\BB(u,\{u,u-c\})=\frac{T_{1,3}(u)\,T_{2,3}(u-c)|0\rangle}
{2\,\lambda_3(u)\,\lambda_3(u-c)}\,.
\end{equation}

When $v_2\to u+c$ the Bethe vector \r{forA} becomes singular due to a 
spurious pole coming from the function $h(u,v_2)^{-1}$. 
However, the combination $h(u,u+c)\,\BB(u,\{u,u+c\})$ remains non-singular 
and can be considered as the Bethe vector 
\begin{equation*}
\BB'(u,\{u,u+c\})=h(u,u+c+\varepsilon)\,\BB(u,\{u,u+c+\varepsilon\})
\Big|_{\varepsilon\to 0}=\frac{T_{1,3}(u)\,T_{2,3}(u+c)|0\rangle}
{\lambda_3(u)\,\lambda_3(u+c)}\,.
\end{equation*}

Expression \r{forA} for the Bethe vector  $\BB(u,\{u,v_2\})$ can be also 
seen from the defining relation \r{T-act-g} for off-shell Bethe vectors.
In the case of $\fs\fl_3$-type off-shell Bethe vectors this relation is 
\begin{equation}\label{ex2}
T_{1,3}(u)\,\BB(\bt^1,\bt^{2})=
\lambda_3(u)\,h(\bt^1,u)\,h(u,\bt^2)\,
\BB(\{\bt^1,u\},\{\bt^{2},u\})\,.
\end{equation}
When the set $\bt^1=\varnothing$ and the cardinality of the set $\bt^2$ 
is equal to 1 ($\bt^2=\{v_2\}$) the  relation \r{ex2} 
becomes 
\begin{equation}\label{ex0}
T_{1,3}(u)\,\BB(\varnothing,v_2)=
\lambda_3(u)\,h(u,v_2)\,
\BB(u,\{u,v_2\})\,,
\end{equation}
where the Bethe vector $\BB(\varnothing,v_2)$ is equal  
 to
\begin{equation}\label{ex3}
\BB(\varnothing,v_2)=\frac{1}{\lambda_3(v_2)}\,T_{2,3}(v_2)\,|0\rangle\,.
\end{equation}
Substituting \r{ex3} to \r{ex0} one gets \r{forA}.

Using methods developed in \cite{HLPRS17} one can get  
\r{BVAex} calculating the projection of the ordered product of simple root 
currents $\Pfp(F_2(v)F_1(u_1)F_1(u_2))$. 
Rather than doing this calculation, we  present instead the 
Bethe vector $\BB(u,\{u,v_2\})$ \r{forA}
 according to its definition \r{BVg} and using the projection method
\begin{equation}\label{cal1}
\begin{split}
\BB(u,\{u,v_2\})&=
\left.\frac{\Pfp\sk{\cF_{2,1}(\{v_1,v_2\},u)}|0\rangle}{h(v_1,v_2)\,h(v_2,v_1)}
\right|_{u=v_1}=
\frac{1}{h(u,v_2)}\ \Pfp\Big(F_{3,2}(v_2)\,F_{3,1}(u)\Big)|0\rangle\,.
\end{split}
\end{equation}
The fact that the product of the currents $F_{3,2}(v_2)\,F_{3,1}(u)$ 
under projection in \r{cal1}
is well defined was explained in \cite{LPR-SR-CC}.

To calculate the
projection in \r{cal1}  we note first that \footnote{We remind that $\Pfp\Big(\Pfm(\F)\,\F'\Big)=0$.} 
\begin{equation}\label{pro1}
\Pfp\Big(F_{3,2}(v)\,F_{3,1}(u)\Big)=\Pfp\Big(\FF^+_{3,2}(v)\,F_{3,1}(u)\Big)\,,
\end{equation}
where 
\begin{equation}\label{pro2}
\FF^+_{3,2}(v)=\Pfp\Big(F_{2}(v)\Big)=\sum_{m\geq 0}F_{2}[m](v/c)^{-m-1}\,.
\end{equation}
The commutation  relation 
\begin{equation}\label{3231a}
F_2(z)\,F_{3,1}(u)=f(u,z)\,F_{3,1}(u)\,F_2(z)
\end{equation}
between the simple root and composed currents is  implied by  the Serre relation
\r{SR-g} (see section~4.2 in \cite{LPR-SR-CC}). Taking into account 
that the function $f(u,z)$ in the commutation relation \r{3231a}
is a series with respect to the powers of $z/u$ and 
looking  the coefficient of $(z/c)^{-a-1}$ for $a\geq0$ 
 one gets 
\begin{equation*}
F_2[a]\,F_{3,1}(u)=F_{3,1}(u)\sk{F_2[a]+\sum_{b=0}^{\infty} 
F_2[a+b](u/c)^{-b-1}}\,.
\end{equation*}
Multiplying this equality by $(v/c)^{-a-1}$ and summing over $a\geq 0$ 
one gets 
\begin{equation}\label{pro4}
\FF^+_{3,2}(v)\,F_{3,1}(u)=F_{3,1}(u)\sk{f(u,v)\ \FF^+_{3,2}(v)-
g(u,v)\ \FF^+_{3,2}(u)}\,,
\end{equation}
that is to say
\begin{equation}\label{pro5}
\Pfp\Big(F_{3,2}(v)\,F_{3,1}(u)\Big)
=\FF^+_{3,1}(u)\sk{f(u,v)\ \FF^+_{3,2}(v)-
g(u,v)\ \FF^+_{3,2}(u)}\,.
\end{equation}
In \r{pro5} we used the property 
$\Pfp(\F_1\cdot\Pfp(\F_2))=\Pfp(\F_1)\cdot\Pfp(\F_2)$ 
and the relation  $\Pfp(F_{3,1}(u))=\FF^+_{3,1}(u)$ (see proposition~\ref{inv-DF}).

In order to verify relation \r{forA} one has to replace the action 
of the Gaussian coordinates in \r{pro5} on the vacuum vector $|0\rangle$ 
by the action of the monodromy entries using Gaussian decomposition \r{Gauss}
\begin{equation}\label{pro7}
T_{2,3}(v)=\FF^+_{3,2}(v)\,k^+_3(v),\quad T_{1,3}(u)=\FF^+_{3,1}(u)\,k^+_3(u)\,.
\end{equation}
The commutation relation
\begin{equation}\label{pro8}
k^+_3(u)\,\FF^+_{3,2}(v_2)\,k^+_3(u)^{-1}=f(u,v_2)\ \FF^+_{3,2}(v_2)-
g(u,v_2)\ \FF^+_{3,2}(u)
\end{equation}
allows to obtain 
\begin{equation}\label{pro11}
\Pfp\Big(F_2(v_2)\,F_{3,1}(u)\Big)|0\rangle=\frac{T_{1,3}(u)\,T_{2,3}(v_2)|0\rangle}
{\lambda_3(u)\,\lambda_3(v_2)}
\end{equation}
which together with \r{cal1} verifies \r{forA} in the framework of the projection
method.

\end{document}